\documentclass[11pt,a4paper]{article}

\usepackage{amsmath,amsfonts,amssymb,amsthm}
\usepackage{mathbbol}
\usepackage{amsthm}
\usepackage{graphicx,color}
\usepackage{boxedminipage}
\usepackage[linesnumbered, boxed, algosection,noline]{algorithm2e}
\usepackage{framed}
\usepackage{thmtools}
\usepackage{thm-restate}
\usepackage{xspace}
\usepackage{vmargin}
\setmarginsrb{1in}{1in}{1in}{1in}{0pt}{0pt}{0pt}{6mm}
\usepackage{todonotes}
 \usepackage[pdftex, plainpages = false, pdfpagelabels, 
                 bookmarks=false,
                 bookmarksopen = true,
                 bookmarksnumbered = true,
                 breaklinks = true,
                 linktocpage,
                 pagebackref,
                 colorlinks = true,  
                 linkcolor = blue,
                 urlcolor  = blue,
                 citecolor = red,
                 anchorcolor = green,
                 hyperindex = true,
                 hyperfigures
                 ]{hyperref} 
 \usepackage{xifthen}
 \usepackage{tabularx}

 \usetikzlibrary{calc}

\newcommand{\GF}{{\sf GF}$(2)$\xspace}
\newcommand{\rank}{{\rm rank}\xspace}
\newcommand{\GFrank}{{\sf GF}$(2)$-{\rm rank}\xspace}
\newcommand{\hdist}{d_H}

\newcommand{\specialcell}[2][c]{%
  \begin{tabular}[#1]{@{}c@{}}#2\end{tabular}}

\DeclareMathOperator{\operatorClassNP}{{\sf NP}}
\newcommand{\classNP}{\ensuremath{\operatorClassNP}}
\DeclareMathOperator{\operatorClassCoNP}{{\sf coNP}}
\newcommand{\classCoNP}{\ensuremath{\operatorClassCoNP}}
\DeclareMathOperator{\operatorClassFPT}{{\sf FPT}\xspace}
\newcommand{\classFPT}{\ensuremath{\operatorClassFPT}\xspace}
\DeclareMathOperator{\operatorClassW}{{\sf W}}
\newcommand{\classW}[1]{\ensuremath{\operatorClassW[#1]}}

\newcommand{\bfA}{\mathbf{A}} 
\newcommand{\bfB}{\mathbf{B}} 
 
\newcommand{\bfD}{\mathbf{D}} 
 
\newcommand{\bfP}{\mathbf{P}} 
\newcommand{\bfU}{\mathbf{U}}
\newcommand{\bfV}{\mathbf{V}}  
\newcommand{\bfa}{\mathbf{a}} 
\newcommand{\bfb}{\mathbf{b}} 
\newcommand{\bfc}{\mathbf{c}} 
\newcommand{\bfd}{\mathbf{d}} 
\newcommand{\bfe}{\mathbf{e}}
\newcommand{\bfs}{\mathbf{s}} 
\newcommand{\bfS}{\mathbf{S}} 
\newcommand{\bfx}{\mathbf{x}} 
\newcommand{\bfX}{\mathbf{X}} 
\newcommand{\bfY}{\mathbf{Y}} 
\newcommand{\bfy}{\mathbf{y}}

\newcommand{\Oh}{\mathcal{O}}

\newtheorem{theorem}{Theorem}
\newtheorem{lemma}{Lemma}
\newtheorem{claim}{Claim}[section]

\newtheorem{definition}{Definition}
\newtheorem{observation}{Observation}
\newtheorem{proposition}{Proposition}
\newtheorem{redrule}{Reduction Rule}

\newtheorem{reduction}{Reduction Rule}[section]

\newcommand{\yes}{{yes}}
\newcommand{\no}{{no}}
\newcommand{\yesinstance}{\yes-instance\xspace}
\newcommand{\noinstance}{\no-instance\xspace}
\newcommand{\yesinstances}{\yes-instances\xspace}

\newcommand{\pname}{\textsc}
\newcommand{\ProblemFormat}[1]{\pname{#1}}
\newcommand{\ProblemIndex}[1]{\index{problem!\ProblemFormat{#1}}}
\newcommand{\ProblemName}[1]{\ProblemFormat{#1}\ProblemIndex{#1}{}\xspace}

\newcommand{\probClust}{\ProblemName{Binary $r$-Means}}
\newcommand{\probAtMostClust}{\ProblemName{Binary $r$-Means}}
\newcommand{\probClustSelect}{\ProblemName{Cluster Selection}}
\newcommand{\probPClust}{\ProblemName{$\bfP$-Matrix Approximation}}
\newcommand{\probBFact}{\ProblemName{Low Boolean-Rank Approximation}}
\newcommand{\probFact}{\ProblemName{Low GF(2)-Rank Approximation}}
\newcommand{\probPClustExt}{\ProblemName{Extendable $\bfP$-Matrix Approximation}}
\newcommand{\probCons}{\ProblemName{Consensus String with Outliers}}



\newlength{\RoundedBoxWidth}
\newsavebox{\GrayRoundedBox}
\newenvironment{GrayBox}[1]%
   {\setlength{\RoundedBoxWidth}{.93\textwidth}
    \def\boxheading{#1}
    \begin{lrbox}{\GrayRoundedBox}
       \begin{minipage}{\RoundedBoxWidth}}%
   {   \end{minipage}
    \end{lrbox}
    \begin{center}
    \begin{tikzpicture}%
       \node(Text)[draw=black!20,fill=white,rounded corners,%
             inner sep=2ex,text width=\RoundedBoxWidth]%
             {\usebox{\GrayRoundedBox}};
        \coordinate(x) at (current bounding box.north west);
        \node [draw=white,rectangle,inner sep=3pt,anchor=north west,fill=white] 
        at ($(x)+(6pt,.75em)$) {\boxheading};
    \end{tikzpicture}
    \end{center}}     

\newenvironment{defproblemx}[2][]{\noindent\ignorespaces%
                                \FrameSep=6pt%
                                \parindent=0pt%
                \vspace*{-1.5em}
                \ifthenelse{\isempty{#1}}{%
                  \begin{GrayBox}{\textsc{#2}}%
                }{%
                  \begin{GrayBox}{\textsc{#2} parameterized by~{#1}}%
                }
                \begin{tabular*}{\textwidth}{@{\hspace{.1em}} >{\itshape} p{1.8cm} p{0.8\textwidth} @{}}%
            }{
                \end{tabular*}%
                \end{GrayBox}%
                \ignorespacesafterend
            }

\newcommand{\defproblema}[3]{
  \begin{defproblemx}{#1}
    Input:  & #2 \\
    Task: & #3
  \end{defproblemx}
}%

\begin{document}

\title{Parameterized  Low-Rank Binary Matrix Approximation \thanks{The research leading to these results have  been supported by the Research Council of Norway via the projects ``CLASSIS'' and ``MULTIVAL".}
}

\author{
Fedor V. Fomin\thanks{
Department of Informatics, University of Bergen, Norway.} \addtocounter{footnote}{-1}
\and
Petr A. Golovach\footnotemark{} \addtocounter{footnote}{-1}
\and 
Fahad Panolan\footnotemark{}
}

\date{}

\maketitle

\begin{abstract}
We provide a number of algorithmic results for  the following family of problems: For a given binary $m\times n$ matrix $\bfA$ and integer $k$, decide whether there is  a ``simple'' binary matrix $\bfB$ which differs  from $\bfA$ in at most $k$ entries. For an integer $r$, the 
  ``simplicity'' of $\bfB$  is characterized  as follows. 
\begin{itemize}
\item 
\probClust: Matrix $\bfB$ has   at most $r$ different columns.  This problem is known to be NP-complete already for $r=2$. We show that the problem is solvable in time $2^{\Oh(k\log k)}\cdot(nm)^{\Oh(1)}$ and thus is fixed-parameter tractable parameterized by $k$. We prove that the problem admits a polynomial kernel when parameterized by $r$ and $k$ but 
it has no polynomial kernel when parameterized by $k$ only unless $\classNP\subseteq\classCoNP/{\rm poly}$.
 We also complement these result by showing that when being parameterized by $r$ and $k$, the problem admits an algorithm of running time 
$ 2^{\Oh(r\cdot \sqrt{k\log{(k+r)}})}(nm)^{\Oh(1)}$, which is subexponential in $k$ for $r\in \Oh(k^{1/2 -\varepsilon})$ for any $\varepsilon>0$.  

\item 
\probFact: Matrix $\bfB$  is of   \GF-rank  at most $r$. This problem is known to be NP-complete already for $r=1$. It also known to be W[1]-hard when parameterized by $k$.
Interestingly, when parameterized by $r$ and $k$, the problem is not only fixed-parameter tractable, but it is solvable in  time   $ 2^{\Oh(r^{ 3/2}\cdot \sqrt{k\log{k}})}(nm)^{\Oh(1)}$, which is subexponential in $k$.

\item 
\probBFact: Matrix $\bfB$ is of 
Boolean rank at most $r$. 
The problem  is known to be NP-complete  for $k=0$ as well as  for $r=1$. 
We show that it is solvable in subexponential in $k$ time $2^{\Oh(r2^r\cdot \sqrt{k\log k})}(nm)^{\Oh(1)}$.    
\end{itemize}
\end{abstract}

\section{Introduction}\label{sec:intro}
In this paper we consider the following generic problem. Given a binary  $m\times n$ matrix, that is a matrix with entries from domain $\{0,1\}$, 
\begin{equation*}\bfA=
\begin{pmatrix}
a_{11}&a_{12}& \ldots& a_{1n}\\
a_{21}&a_{21}& \ldots& a_{2n}\\
\vdots & \vdots & \ddots & \vdots\\ 
a_{m1}&a_{m2} &\ldots &a_{mn}
\end{pmatrix}=(a_{ij})\in \{0,1\}^{m\times n}, 
\end{equation*}
the task is to find a ``simple'' binary  $m\times n$ matrix $\bfB$  which  approximates $\bfA$ subject to some specified constrains. One of the   most widely studied error measures is the \emph{Frobenius norm}, which for a matrix $\bfA$ is defined as
\begin{equation*}
\|\bfA\|_F=\sqrt{\sum_{i=1}^m\sum_{j=1}^n|a_{ij}|^2}.
\end{equation*}
Here the sums are taken over $\mathbb{R}$. Then for a given nonnegative integer $k$, we want to decide whether there is a matrix $\bfB$ with certain properties such that 
\[\|\bfA-\bfB\|_F^2\leq k.\]

We consider the binary matrix approximation problems when for a given integer $r$,  the  approximation binary matrix $\bfB$ 
\begin{itemize}
\item[(A1)] has at most $r$ pairwise-distinct columns, 
\item[(A2)] is of   \GF-rank  at most $r$, and 
\item[(A3)] is of Boolean rank at most $r$.
\end{itemize}
Each of these variants is very well-studied. Before defining each of the problems formally and providing an overview of the relevant results, the following observation is in order. Since we approximate a binary matrix by a binary matrix, in this case minimizing the Frobenius norm of $\bfA-\bfB$ is equivalent to minimizing the $\ell_0$-norm of $\bfA-\bfB$,  where  the measure $\|\bfA\|_0$ is the number of non-zero entries of matrix $\bfA$.
We also will be using another  equivalent way of measuring the quality of approximation of a binary matrix $\bfA$ by a binary matrix $\bfB$ by taking the sum of the Hamming distances between their columns. 
Let us recall that the \emph{Hamming distance} between two vectors $\bfx, \bfy\in\{0,1\}^m$, where $\bfx=(x_1,\ldots,x_m)^\intercal$ and $\bfy=(y_1,\ldots,y_m)^\intercal$, is $\hdist(\bfx,\bfy)=\sum_{i=1}^m |x_i-y_i|$ or, in   words, the number of positions $i\in\{1,\ldots,m\}$ where $x_i$ and $y_i$ differ. 
Then for   binary 
$m\times n$ matrix $\bfA$   with columns $\bfa^1,\ldots,\bfa^n$ and  matrix $\bfB$   with columns $\bfb^1,\ldots,\bfb^n$, we define 
\begin{equation*}
d_H(\bfA,\bfB)= \sum_{i=1}^n d_H(\bfa^i, \bfb^i).
\end{equation*}
 In other words, $d_H(\bfA,\bfB)$ is the number of positions with different entries in matrices $\bfA$ and $\bfB$. 
Then we  have the following.
\begin{equation}\label{eqn:FrHam}
\|\bfA-\bfB\|_F^2= \|\bfA-\bfB\|_0=d_H(\bfA,\bfB)=\sum_{i=1}^n d_H(\bfa^i, \bfb^i).
\end{equation}

\medskip\noindent\textbf{Problem (A1): \probClust.}
By \eqref{eqn:FrHam}, the problem  of approximating a binary $m\times n$ matrix $\bfA$ by a binary  $m\times n$ matrix $\bfB$ with at most $r$ different columns (problem (A1)) is equivalent to the following clustering problem. For given a set of $n$ binary $m$-dimensional vectors $\bfa^1,\ldots,\bfa^n$ (which constitute the columns of matrix $\bfA$) and a positive integer $r$,  \probClust aims to partition the vectors in at most $r$ clusters,  as to minimize the sum of within-clusters sums of Hamming distances to their binary means. More formally,

\defproblema{\probAtMostClust}%
{An $m\times n$ matrix $\bfA$ with columns $(\bfa^1,\ldots,\bfa^n)$, a positive integer $r$ and a nonnegative integer $k$.}%
{Decide whether there is 
a positive integer $r'\leq r$, a partition $\{I_1,\ldots,I_{r'}\}$ of $\{1,\ldots,n\}$ and vectors $\bfc^1,\ldots,\bfc^{r'}\in\{0,1\}^m$ such that 
$$\sum_{i=1}^{r'}\sum_{j\in I_i}\hdist(\bfc^i,\bfa^j)\leq k?$$
}

To see the equivalence of \probAtMostClust and problem (A1), it is sufficient to observe that the pairwise different columns of an approximate matrix $\bfB$ such that $\|\bfA-\bfB\|_0\leq k$ can be used as vectors  $\bfc^1,\ldots,\bfc^{r'}$, $r'\leq r$.
As far as the mean vectors are selected, a partition  of columns of $\bfA$ can be obtained by assigning each column-vector $\bfa^i$ to its closest mean vector  $\bfc^j$ (ties breaking arbitrarily).  Then for such clustering the total sum of distances from vectors within cluster to their centers  does not exceed $k$. Similarly, solution to \probAtMostClust can be used as columns (with possible repetitions) of matrix $\bfB$ such that   $\|\bfA-\bfB\|_0\leq k$. For that  we put $\bfb^i=\bfc^j$, where  $\bfc^j$ is the closest  vector to $\bfa^i$.

This problem was introduced by Kleinberg, Papadimitriou,  and Raghavan \cite{KleinbergPR04} as one of the examples of segmentation problems.  Approximation algorithms for optimization versions of this  problem  were given by 
Alon and Sudakov \cite{AlonS99} and Ostrovsky  and Rabani \cite{OstrovskyR02},  who referred to it as clustering in the Hamming cube. 
In bioinformatics, the case when $r=2$   is known  under the name \textsc{Binary-Constructive-MEC} (Minimum  Error  Correction) and was studied as a model for the 
\textsc{Single Individual
Haplotyping} problem \cite{CilibrasiIK07}. 
Miettinen et al. \cite{MiettinenMGDM08} studied this problem under the name
 \textsc{Discrete Basis Partitioning Problem}.

\probAtMostClust can be seen as  a discrete variant of the well-known \textsc{$k$-Means Clustering}.   (Since in problems (A2) and (A3) we use   $r$ for the rank of the approximation matrix, we also use $r$ in (A1) to denote the number of clusters which is commonly denoted by $k$ in the literature on means clustering.) This problem has been studied thoroughly, particularly in the areas of computational geometry and machine learning.  We refer to \cite{DBLP:journals/jacm/AgarwalHV04,DBLP:conf/stoc/BadoiuHI02,KumarSS10}  for further references to the works  on \textsc{$k$-Means Clustering}.

\medskip\noindent\textbf{Problem (A2): \probFact.}  
Let  $\bfA$ be a $m\times n$ binary matrix.  
In this case we view the elements of   $\bfA$ as elements of \GF, the Galois field of two elements. Then the \GFrank of $\bfA$ is the minimum $r$ such that  $\bfA =\bfU \cdot \bfV$, where  $\bfU$ is $m\times r$ and $ \bfV$ is $r\times n$ binary matrices, and arithmetic operations are over \GF. Equivalently, this is the minimum number of binary vectors, such that every column (row)  of $\bfA$ is a linear combination (over \GF) of 
 these vectors. 
Then (A2) is the following problem.

\defproblema{\probFact}%
{An $m\times n$-matrix $\bfA$ over \GF, a positive integer $r$ and a nonnegative integer $k$.}%
{Decide whether there is a binary  $m\times n$-matrix $\bfB$   with \GFrank$(\bfB)\leq r$ such that $\|\bfA-\bfB\|_F^2\leq k$.
}
\probFact  arises naturally in applications involving binary data sets and 
serves as an important tool in dimension reduction for high-dimensional data sets with binary attributes, see  \cite{DanHJWZ15,Jiang2014,GutchGYT12,Koyuturk2003,PainskyRF16,Shen2009,Yeredor11} for further references and  numerous applications of the problem.

\probFact can be rephrased as a special variant (over \GF)   of the problem finding the rigidity of a matrix. 
(For a target rank $r$, the {\em rigidity} of a matrix $A$ over a field $\mathbb{F}$ is the minimum Hamming distance between $A$ and a matrix of rank at most $r$.) Rigidity is a classical concept in Computational Complexity Theory studied due to its connections with lower bounds for arithmetic circuits \cite{grigoRigid,grigoRigidTra,Valiant77,raz89}. We refer 
to~\cite{Lokam2009} for an extensive  survey on this topic. 
 
 \probFact is also a special case of a general class of problems approximating a matrix by a matrix with a small non-negative rank.  Already \textsc{Non-negative Matrix Factorization} (NMF) is a nontrivial problem and it appears in many settings.   In particular, in machine learning,  approximation by a low non-negative rank matrix  has gained extreme popularity after  the influential article in Nature by  Lee and Seung  \cite{lee1999learning}.   NMF is an ubiquitous problem and besides machine learning, it has been independently introduced and studied in combinatorial optimization \cite{FioriniMPTW15,Yannakakis91}, and communication complexity   \cite{AhoUY83,LovaszS88}. An extended overview of  applications of NMF in statistics, quantum mechanics, biology, economics, and chemometrics,  can be found in the work of Cohen and Rothblum  \cite{cohen1993nonnegative} and  recent books \cite{cichocki2009nonnegative,naiknon,book:1292686}. 

\medskip\noindent\textbf{Problem (A3): \probBFact.} 
 Let $\bfA$ be a binary $m\times n$ matrix. 
 This time we view the elements of $\bfA$ as \emph{Boolean} variables. 
The \emph{Boolean rank} of $\bfA$ is the minimum $r$ such that $\bfA=\bfU\wedge \bfV$ for a Boolean $m\times r$ matrix $\bfU$ and a Boolean $r\times n$ matrix $\bfV$, where the product is Boolean, that is,  the logical $\wedge$ plays the role of multiplication and $\vee$ the role of sum. Here  $0\wedge 0=0$, $0 \wedge 1=0$, $1\wedge 1=1$ , $0\vee0=0$, $0\vee1=1$, and  $1\vee 1=1$.  
Thus the  matrix product is over the Boolean semi-ring $({0, 1}, \wedge, \vee)$. This can be equivalently expressed
as the  normal matrix product with addition defined as $1 + 1 =1$. Binary matrices equipped with such algebra are called \emph{Boolean
matrices}. 
Equivalently, $\bfA =(a_{ij})\in \{0,1\}^{m\times n}$ has the Boolean rank $1$ if $\bfA=\bfx^\intercal\wedge \bfy$, where $\bfx =(x_1, x_2, \dots, x_m)\in\{0,1\}^m$ and $\bfy=(y_1, y_2, \dots, y_n)\in\{0,1\}^n$ are nonzero vectors and the product is  Boolean,
that is,   $a_{ij}=x_i \wedge y_j$. 
Then the Boolean rank of $\bfA$ is the minimum integer $r$ such that $\bfA=\bf\bfa^{(1)}\vee\cdots\vee \bf\bfa^{(r)}$,  where $\bf\bfa^{(1)},\ldots, \bf\bfa^{(r)}$  are matrices of Boolean rank 1; 
zero matrix is the unique matrix with the Boolean rank $0$. Then \probBFact is defined as follows.

\defproblema{\probBFact}%
{A Boolean $m\times n$ matrix $\bfA$, a positive integer $r$ and a nonnegative integer $k$.}%
{Decide whether there is a Boolean $m\times n$ matrix $\bfB$ of Boolean rank at most $r$ such that  $d_H(\bfA,\bfB)\leq k$. 
}

For $r=1$ \probBFact coincides with \probFact but for $r>1$ these are different problems.
 
 Boolean low-rank approximation  has  attracted much attention, especially in the data mining and knowledge discovery communities. 
In data mining, matrix decompositions are often used to produce concise representations of data. 
Since much of the real data such as word-document data  is binary or even Boolean in nature, 
Boolean low-rank approximation could provide a deeper insight into 
the semantics associated with the original matrix. There is a big body of work done on \probBFact, see e.g.  \cite{Bartl2010,BelohlavekV10,DanHJWZ15,LuVAH12,MiettinenMGDM08,DBLP:conf/kdd/MiettinenV11,DBLP:conf/icde/Vaidya12}. In the literature the problem appears under different names like \textsc{Discrete Basis Problem} \cite{MiettinenMGDM08} or 
\textsc{Minimal Noise Role Mining Problem} \cite{VaidyaAG07,LuVAH12,Mitra:2016}.

\medskip\noindent\textbf{\probPClust.}  While at   first glance  \probFact and \probBFact look very similar, algorithmically the latter problem is  more challenging. The fact that \GF is a field allows to play with different equivalent definitions of rank like row rank and column ranks. We exploit this strongly in our algorithm for \probFact.  For \probBFact the matrix product is over the Boolean semi-ring and nice properties of the \GF-$\rank$ cannot be used 
 here  (see, e.g.~\cite{Guterman08}). Our algorithm for  \probBFact is based on solving an auxiliary \probPClust problem, where the task is  to approximate a matrix $\bfA$ by a matrix $\bfB$ whose block structure is defined  by a given pattern matrix $\bfP$. It appears, that \probPClust is also an interesting problem on its own. 
 
More formally, 
  let $\bfP=(p_{ij})\in \{0,1\}^{p\times q}$ be a binary $p\times q$ matrix.  
   We say that a binary $m\times n$  matrix 
    $\bfB=(b_{ij})\in \{0,1\}^{m\times n}$ is a \emph{$\bfP$-matrix} if there is a partition $\{I_1,\ldots,I_p\}$ of $\{1,\ldots,m\}$ and a partition $\{J_1,\ldots,J_q\}$ of $\{1,\ldots,n\}$ such that for every $i\in\{1,\ldots,p\}$, $j\in\{1,\ldots,q\}$,  $s\in I_i$ and $t\in J_j$, $b_{st}=p_{ij}$. In words,   the columns and rows of $\bfB$ can be permuted such that the block structure  of the resulting matrix is defined by $\bfP$.

\defproblema{\probPClust}%
{An $m\times n$ binary matrix $\bfA$, a pattern  binary matrix $\bfP$ and a nonnegative integer $k$.}%
{Decide whether there is an $m\times n$ $\bfP$-matrix $\bfB$ such that $\|\bfA-\bfB\|_F^2\leq k$.
}

 The notion of $\bfP$-matrix was implicitly defined by Wulff et al. \cite{WulffUB13} as an auxiliary tool for their approximation algorithm for the related monochromatic biclustering problem.  \probPClust is also closely related to the problems arising in 
tiling transaction databases (i.e., binary matrices), where the task is to find a  tiling covers of a given binary matrix with a small number of submatrices full of 1s, see
\cite{DBLP:conf/dis/GeertsGM04}.

  Since \probFact remains  \classNP-complete for $r=1$
  ~\cite{GillisV15}, we have that  \probPClust is \classNP-complete already  for very simple pattern matrix $P=
\begin{pmatrix}
0&0\\
0&1
\end{pmatrix}.
$

\subsection{Related work}In this subsection we give an overview of previous related algorithmic and complexity results for problems (A1)--(A3), as well as related problems. Since each of the problems has many practical applications, there is a tremendous amount of literature on heuristics and implementations.   In this overview we concentrate on known results about algorithms with proven guarantee, with emphasis on parameterized complexity.

\medskip\noindent\textbf{Problem (A1): \probClust.}
\probAtMostClust is trivially solvable in polynomial time for $r=1$, and as was shown by Feige in~\cite{Feige14b}, is \classNP-complete for every  $r\geq 2$.

  PTAS (polynomial time approximation scheme) for optimization variants of  \probClust were developed in  
\cite{AlonS99,OstrovskyR02}.  Approximation algorithms for  more general \textsc{$k$-Means Clustering} is  a thoroughly studied topic
\cite{DBLP:journals/jacm/AgarwalHV04,DBLP:conf/stoc/BadoiuHI02,KumarSS10}.
Inaba et al.~\cite{inaba1994applications} have shown that  the general  
  \textsc{$k$-Means Clustering}  is solvable in time 
  $n^{mr +1}$ (here $n$ is the number of vectors, $m$ is the dimension and $r$ the number of required clusters). We are not aware of any, except the trivial brute-force, exact algorithm for \probClust prior to our work.

\medskip\noindent\textbf{Problem (A2): \probFact.} 
When low-rank approximation matrix $\bfB$ is not required to be binary, then 
the optimal Frobenius norm rank-$r$ approximation of (not necessarily binary) matrix $\bfA$ can be efficiently found via the singular value decomposition
(SVD). This is an extremely well-studied problem and we refer to  surveys for an overview of algorithms for  low rank approximation \cite{RavindrV08,MahonyM11,Woodr14}.  
However, SVD does not guarantee to find an optimal solution in the case when additional structural constrains on  the low-rank approximation matrix $\bfB$ (like being non-negative or  binary) are imposed.

 In fact, most of these constrained variants of low-rank approximation are NP-hard. In particular, 
Gillis and Vavasis ~\cite{GillisV15} and Dan et al. \cite{DanHJWZ15} have shown that \probFact is \classNP-complete  for every $r\geq1$.
Approximation algorithms for the optimization version of  \probBFact were considered in 
 \cite{DBLP:conf/icdm/JiangH13a,Jiang2014,DanHJWZ15,Koyuturk2003,Shen2009,BringmannKW17} among others.

 Most of the known results about the parameterized complexity of the problem follows from the results for \textsc{Matrix Rigidity}.
 Fomin et al.  have proved in \cite{FominLMSZ17}  that  for  every finite field,  
  and  in particular   \GF, \textsc{Matrix Rigidity}
over a finite field  is  $\classW1$-hard being parameterized by $k$. This implies that \probFact is  $\classW1$-hard when parameterized by $k$.  However, when parameterized by $k$ and $r$, the problem becomes fixed-parameter tractable. For  \probFact,  the algorithm from \cite{FominLMSZ17} runs in time 
$2^{\Oh(f(r) \sqrt{k} \log k)} (nm)^{\Oh(1)}$, where $f$ is some function of $r$. While the function $f(r)$ is not specified in 
 \cite{FominLMSZ17}, the algorithm in  \cite{FominLMSZ17} invokes enumeration of all $2^r \times 2^r$ binary matrices of rank $r$, and thus the running time is at least double-exponential in $r$.

 Meesum, Misra, and Saurabh \cite{MeesumMS},  and Meesum and Saurabh \cite{MeesumS16} considered parameterized algorithms for related problems about editing of the adjacencies of a graph (or directed graph) targeting a graph with adjacency matrix of small rank.

\medskip\noindent\textbf{Problem (A3): \probBFact.}  It follows from the rank definitions that a matrix is of Boolean rank $r=1$ if and only if its \GF-rank is $1$.  Thus by the results of Gillis and Vavasis ~\cite{GillisV15} and Dan et al. \cite{DanHJWZ15} \probBFact is \classNP-complete already for $r=1$. 
Lu et al. \cite{DBLP:conf/icde/LuVA08} gave a formulation of  \probBFact as an integer programming problem with exponential number of variables 
and constraints.

While computing  \GF-rank (or rank over any other field) of a matrix can be performed in polynomial time, 
deciding whether the Boolean rank of a given matrix is at most $r$ is already an \classNP-complete problem. 
Thus \probBFact is \classNP-complete already for $k=0$.
This follows from the well-known relation between the Boolean rank and covering edges of a bipartite graph by bicliques \cite{GregoryPJL91}. 
Let us briefly describe this equivalence. For Boolean matrix $\bfA$, let $G_A$ be the corresponding bipartite graph, i.e. the bipartite graph whose biadjacency matrix is $\bfA$. By the 
  equivalent   definition  of the Boolean rank,   $\bfA$ has Boolean rank $r$ if and only if it is the logical disjunction  of $r$ Boolean matrices of rank $1$.  But for every bipartite graph whose biadjacency matrix is a Boolean matrix of rank at most $1$, its edges can be covered by at most one biclique (complete bipartite graph). Thus deciding whether a matrix is of Boolean rank $r$ is exactly the same as deciding whether edges of a bipartite graph can be covered by at most $r$ bicliques. The latter \textsc{Biclique Cover} problem is known to be \classNP-complete~\cite{Orlin77}. 
 \textsc{Biclique Cover} is solvable in time $2^{2^{\Oh(r)}}(nm)^{\Oh(1)}$ ~\cite{GrammGHN08}
  and unless Exponential Time Hypothesis (ETH) fails, it cannot be solved in time $2^{2^{o(r)}}(nm)^{\Oh(1)}$
  \cite{ChandranIK16}.

  For the special case  $r=1$  \probBFact  and $k \leq \|A\|_0 /240$,  Bringmann, Kolev and Woodruff gave an exact algorithm of running time $2^{k/\sqrt{\|A\|_0 }} (nm)^{\Oh(1)}$ \cite{BringmannKW17}. (Let us remind that the $\ell_0$-norm  of a matrix is the number of its non-zero entries.) 
More generally, exact algorithms for NMF were studied by Cohen and Rothblum  in \cite{cohen1993nonnegative}.  Arora et al.  \cite{AroraGKM12} and Moitra~\cite{Moitra16}, who showed that for a fixed value of $r$, NMF is solvable in polynomial time. Related are also the works of  Razenshteyn et al.~\cite{RazenshteynSW16}  on weighted low-rank approximation,  Clarkson and Woodruff~\cite{ClarksonW15}  on robust subspace approximation, and  Basu et al.~ \cite{BasuDL16} on PSD factorization.

Observe that all the  problems studied in this paper could be seen as matrix editing problems. For \probClust, we can assume that $r\leq n$ as otherwise we have a trivial NO-instance. Then the problem asks whether it is possible to edit at most $k$ entries of the input matrix, that is, replace some $0$s by $1$s and  some $1$s by $0$s, in such a way that the obtained matrix has at most $r$ pairwise-distinct columns. Respectively, \probFact asks whether  it is possible to edit at most $k$ entries of the input matrix to obtain a matrix of rank at most $r$.  In \probPClust, we ask whether we can edit at most $k$ elements to obtain a $\bfP$-matrix. A lot of work in graph algorithms has been done on graph editing problems, in particular parameterized subexponential time algorithms were developed for a number of problems, including various cluster editing problems \cite{drange2015fastarxiv,fomin2011subexponential}.

\subsection{Our results and methods}
We study the parameterized complexity of \probClust, \probFact and  \probBFact.  We refer to the recent books of Cygan et al.~\cite{CyganFKLMPPS15} and  Downey and Fellows~\cite{DowneyF13} for  the introduction  to Parameterized Algorithms and Complexity. Our results are summarized in Table~\ref{tabl:compl}.

\begin{table}[ht]
\begin{center}
{\small
\begin{tabular}{|l|c|c|c|}
\hline
& $k$  & $r$ & $k+r$\\
\hline
\probAtMostClust 
&  \specialcell{$2^{\Oh(k\log k)}$,  Thm~\ref{thm:fpt-clust-k}\\ No poly-kernel, Thm~\ref{thm:no-kernel}}  &  \classNP-c  for $r\geq 2$ \cite{Feige14b}&\specialcell{$2^{\Oh(r\cdot \sqrt{k\log{(k+r)}})}$, Thm~\ref{thm:clust-subexp} \\ Poly-kernel, Thm~\ref{thm:kernel}}
\\
\hline
 \textsc{\GF  Approx }  &\classW1-hard \cite{FominLMSZ17} &  \classNP-c  for $r\geq 1$ \cite{GillisV15,DanHJWZ15} &$2^{\Oh(r^{3/2}\cdot \sqrt{k\log{k}})}$, Thm~\ref{thm:fact-subexp} \\
\hline
 \textsc{Boolean Approx }
 & \classNP-c  for $k=0$   \cite{Orlin77} &  \classNP-c  for $r\geq 1$ \cite{GillisV15,DanHJWZ15} &$2^{\Oh(r2^r\cdot \sqrt{k\log k})}$,  Thm~\ref{cor:brank}\\
\hline
                         \end{tabular}
}
\caption{Parameterized complexity of low-rank approximation.  
\textsc{\GF  Approx } stands for \probFact and   \textsc{Bool Approx} for  \probBFact.  We omit  the polynomial factor $(nm)^{\Oh(1)}$ in running times. } \label{tabl:compl}
\end{center}
\end{table}

Our first main result concerns \probClust. We show (Theorem~\ref{thm:fpt-clust-k}) that the problem is solvable in time $2^{\Oh(k\log k)}\cdot(nm)^{\Oh(1)}$. Therefore,  \probClust is \classFPT parameterized by $k$. Since \probFact parameterized by $k$ is \classW1-hard and \probBFact is \classNP-complete for any fixed $k\geq 0$, we find  Theorem~\ref{thm:fpt-clust-k} quite surprising. The proof of Theorem~\ref{thm:fpt-clust-k} is based on a fundamental result of Marx~\cite{Marx08} about the complexity of a problem on strings, namely \textsc{Consensus Patterns}.  
We solve \probClust by constructing a two-stage  FPT Turing  reduction to \textsc{Consensus Patterns}. First, we use   the  color coding technique of Alon, Yuster, and Zwick from~\cite{AlonYZ95}
to reduce \probClust to some special auxiliary problem and then show that this problem can be reduced to \textsc{Consensus Patterns}, and this allows us to apply  the algorithm of    Marx~\cite{Marx08}. We also prove (Theorem~\ref{thm:kernel})  that \probClust admits a polynomial kernel when parameterized by $r$ and $k$. That is, we give a polynomial time algorithm that for a given instance of \probClust  outputs an equivalent instance with $\Oh(k^2+kr)$ columns and $\Oh(k^3+kr) $ rows. 
For parameterization by $k$ only, we show in Theorem~\ref{thm:no-kernel} that  \probClust   has no polynomial kernel    unless $\classNP\subseteq\classCoNP/{\rm poly}$, the standard complexity assumption.

Our second main result concerns \probBFact. As we mentioned above, the problem is NP-complete for $k=0$, as well as for  for $r=1$, and hence  is intractable being parameterized by $k$ or by $r$ only. 
 On the other hand, a simpler \probFact is not only \classFPT parameterized by $k+r$, by~\cite{FominLMSZ17} it is solvable in  time $2^{\Oh(f(r) \sqrt{k} \log k)} (nm)^{\Oh(1)}$, where $f$ is some  function of $r$, and thus is \emph{subexponential} in $k$. It is natural to ask whether a similar complexity behavior could be expected for  \probBFact. Our second main result, Theorem~\ref{cor:brank}, shows that this is indeed the case: \probBFact is solvable in time $2^{\Oh(r2^r\cdot \sqrt{k\log k})}(nm)^{\Oh(1)}$. The proof of this theorem is technical and consists of several steps. 
We first develop a subexponential algorithm for solving auxiliary \probPClust, and then  construct an FPT Turing reduction from \probBFact to \probPClust. 

Let us note that due to the relation of Boolean rank  computation to   \textsc{Biclique Cover}, the result of   \cite{ChandranIK16} implies that 
unless Exponential Time Hypothesis (ETH) fails, \probBFact  cannot be solved in time $2^{2^{o(r)}}f(k)(nm)^{\Oh(1)}$ for any function $f$. Thus the dependence in $r$ in our algorithm cannot be improved significantly unless ETH fails.

Interestingly, the technique developed for solving  \probPClust can be used to obtain algorithms of running times
$ 2^{\Oh(r\cdot \sqrt{k\log{(k+r)}})}(nm)^{\Oh(1)}$ for 
\probClust and
$ 2^{\Oh(r^{ 3/2}\cdot \sqrt{k\log{k}})}(nm)^{\Oh(1)}$  for  \probFact (Theorems~\ref{thm:clust-subexp} and \ref{thm:fact-subexp} 
 respectively). For  \probClust,  Theorems~\ref{thm:clust-subexp}
   provides much better running time than Theorem~\ref{thm:fpt-clust-k} for values of $r\in o((k\log{k})^{1/2})$. 

For \probFact, comparing Theorem~\ref{thm:fact-subexp}
 and the running time $2^{\Oh(f(r) \sqrt{k} \log k)} (nm)^{\Oh(1)}$ from \cite{FominLMSZ17},  let us note that Theorem~\ref{thm:fact-subexp}
  not only slightly improves the exponential dependence in  $k$ by $\sqrt{\log{k}}$; it also drastically improves the exponential dependence in $r$,  from $2^{2^r}$ to   $2^{r^{ 3/2}}$.

The remaining part of the paper is organized as follows. 
In Section~\ref{sec:prelim} we introduce basic notations and obtain some auxiliary results. In Section~\ref{sec:par-k} we show that \probAtMostClust is \classFPT when parameterized by $k$ only.
In Section~\ref{sec:kernel} we discuss kernelization for \probAtMostClust.
 In Section~\ref{sec:subexp} we construct \classFPT algorithms for \probAtMostClust and  \probFact parameterized by $k$ and $r$ that are subexponential in $k$. 
  In Section~\ref{sec:p-clust} we give a subexponential algorithm for \probBFact. We conclude our paper is Section~\ref{sec:conclusion} by stating some open problems.

\section{Preliminaries}\label{sec:prelim}
In this section we introduce the terminology used throughout the paper and obtain some properties of the solutions to our problems. 

All matrices and vectors considered in this paper are assumed to be $(0,1)$-matrices and vectors respectively unless explicitly specified otherwise.  
Let  $\bfA =(a_{ij})\in \{0,1\}^{m\times n}$  be an $m\times n$-matrix. Thus $a_{ij}$,    $i\in\{1,\ldots,m\}$ and $j\in\{1,\ldots,n\}$, are  the elements of $\bfA$. For $I\subseteq\{1,\ldots,m\}$ and $J\subseteq\{1,\ldots,n\}$, we denote by $\bfA[I,J]$ the $|I|\times |J|$-submatrix of $\bfA$ with the elements $a_{ij}$ where $i\in I$ and $j\in J$. We say that two matrices $\bfA$ and $\bfB$ are \emph{isomorphic} if $\bfB$ can be obtained from $\bfA$ by permutations of rows and columns. We use ``$+$'' and ``$\sum$'' to denote sums and summations over $\mathbb{R}$, and we use ``$\oplus$'' and 
``$\bigoplus$'' for 
sums and 
summations over \GF.

We also consider string of symbols.
For two strings $a$ and $b$, we denote by $ab$ their \emph{concatenation}. For a positive integer $k$, $a^k$ denotes the concatenation of $k$ copies of $a$; $a^0$ is assumed to be the empty string.
Let $a=a_1\cdots a_\ell$ be a string over an alphabet $\Sigma$. Recall that a string $b$ is said to be a \emph{substring} of $a$ if $b=a_{h}a_{h+1}\cdots a_t$ for some $1\leq h\leq t\leq\ell$; we write that $b=a[h..t]$ in this case.
Let $a=a_1\cdots a_\ell$ and $b=b_1\cdots b_\ell$ be strings of the same length $\ell$ over  $\Sigma$. Similar to the the definition of Hamming distance between two $(0,1)$-vectors, the \emph{Hamming distance} $\hdist(a,b)$ between two strings is defined as the number of position $i\in\{1,\ldots,\ell\}$ where the strings differ. 
We would like to mention that for Hamming distance (for vectors and strings), the {\em triangular inequality} holds. That is,  
for any three strings $a,b,c$ of length $n$ each, $\hdist(a,c)\leq \hdist(a,b)+\hdist(b,c)$.

\subsection{Properties of  \probAtMostClust}\label{propMOSTCLUST}
Let $(\bfA,r,k)$ be an instance of \probAtMostClust where $\bfA$ is a matrix with columns $(\bfa^1,\ldots,\bfa^n)$.   
We say that a partition $\{I_1,\ldots,I_{r'}\}$ of $\{1,\ldots,n\}$ for $r'\leq r$ is a \emph{solution} for $(\bfA,r,k)$ if there are vectors $\bfc^1,\ldots,\bfc^{r'}\in\{0,1\}^m$ such that $\sum_{i=1}^{r'}\sum_{j\in I_i}\hdist(\bfc^i,\bfa^j)\leq k$. We say that each $I_i$ or, equivalently, the multiset  of columns $\{\bfa^j\mid j\in I_i\}$ (some columns could be the same) is a \emph{cluster} and call $\bfc^i$ the \emph{mean} of the cluster. 
Observe that given a cluster $I\subseteq \{1,\ldots,n\}$, one can easily compute an optimal mean $\bfc=(c_1,\ldots,c_m)^\intercal$ as follows. Let $\bfa^j=(a_{1j},\ldots,a_{mj})^{\intercal}$ for $j\in\{1,\ldots,n\}$. For each $i\in \{1,\ldots,m\}$, consider the multiset $S_i=\{a_{ij}\mid j\in I\}$ and put $c_i=0$ or $c_i=1$ according to the \emph{majority} of elements in $S_i$, that is, $c_i=0$ if at least half of the elements in $S_i$ are $0$s and $c_i=1$ otherwise. We refer to this construction of $\bfc$ as the \emph{majority rule}.

In the opposite direction, given a set of means $\bfc^1,\ldots,\bfc^{r'}$, we can construct clusters $\{I_1,\ldots,I_{r'}\}$ as follows: for each column $\bfa^j$, find the closest $\bfc^i$, $i\in\{1,\ldots,r'\}$, such that $\hdist(\bfc^i,\bfa^j)$ is minimum and assign  $j$ to $I_i$. Note that this procedure does not guarantee that all clusters are nonempty but we can simply delete empty clusters.  Hence, we can define a solution as a set of means $C=\{\bfc^1,\ldots,\bfc^{r'}\}$. These arguments also imply the following observation.

\begin{observation}\label{obs:eq}
The task of \probAtMostClust can equivalently be stated as follows: decide whether there exist a positive integer $r'\leq r$  and vectors $\bfc^1,\ldots,\bfc^{r'}\in\{0,1\}^m$ such that 
$\sum_{i=1}^n\min\{\hdist(\bfc^j,\bfa^i)\mid 1\leq j\leq r'\}\leq k$.
\end{observation}

\begin{definition}[Initial cluster and regular partition]
Let $\bfA$ be an  $m\times n$-matrix  with  columns $\bfa^1,\ldots,\bfa^n$. An \emph{initial cluster}  is an inclusion maximal set $I\subseteq\{1,\ldots,n\}$ such that all the columns in the multiset $\{\bfa^j\mid j\in I\}$ are equal.  

We say that a partition $\{I_1,\ldots,I_{r'}\}$ of the columns of matrix $\bfA$ 
 is \emph{regular} if for every initial cluster $I$, there is $i\in\{1,\ldots,r'\}$ such that $I\subseteq I_i$.
 \end{definition}
By the definition of the regular partition, every initial cluster of $\bfA$ is in some set $I_i$ but the set $I_i$ may contain many  initial clusters.

\begin{lemma}\label{lem:init}
Let $(\bfA,r,k)$ be a \yesinstance of \probAtMostClust. Then there is a solution $\{I_1,\ldots,I_{r'}\}$, $r'\leq r$  
which is regular (i.e, for any initial cluster $I$ of $\bfA$, there is $i\in \{1,\ldots,r'\}$ such that $I\subseteq I_i$). 
\end{lemma}

\begin{proof}
Let $\bfa^1,\ldots,\bfa^n$ be the columns of $\bfA$.
By Observation~\ref{obs:eq}, there are vectors $\bfc^1,\ldots,\bfc^{r'}$ for some $r'\leq r$ such that 
$\sum_{i=1}^{n}\min\{\hdist(\bfc^j,\bfa^i)~|~1\leq j\leq r'\}\leq k$. Once we have the vectors $\bfc^1,\ldots,\bfc^{r'}$, a solution 
can be obtained by assigning each vector $\bfa^i$ to a closest vector in $\{\bfc^1,\ldots,\bfc^{r'}\}$. 
This implies the conclusion of the lemma.   
\end{proof}

\subsection{Properties of  \probFact}
For \probFact, we need the following folklore observation. We provide a proof for completeness.

\begin{observation}\label{obs:rank}
Let $\bfA$ be a matrix over \GF{} with $\rank(\bfA)\leq r$. Then $\bfA$ has at most $2^r$ pairwise-distinct columns and at most $2^r$ pairwise-distinct rows.
\end{observation}

\begin{proof}
We show the claim for columns; the proof for the rows is similar in arguments to that of the case of columns. Assume that $\rank(A)=r$ and let $\bfe_1,\ldots,\bfe_r$ be a basis of the column space of $\bfA$. Then  every column $\bfa^i$ of $\bfA$ is a linear combination of $\bfe_1,\ldots,\bfe_r$. Since $\bfA$ is a matrix over \GF, it implies that for every columns $\bfa^i$, there is $I\subseteq\{1,\ldots,r\}$ such that $\bfa^i=\bigoplus_{j\in I}\bfe_j$. As the number of distinct subsets of $\{1,\ldots,r\}$ is $2^r$, the claim follows.
\end{proof}

By making use of  Observation~\ref{obs:rank},  we can reformulate  \probFact as follows: given an $m\times n$ matrix $\bfA$ over \GF{} with the columns $\bfa^1,\ldots,\bfa^n$, a positive integer $r$ and a nonnegative integer $k$, we ask whether there is a positive integer $r'\leq 2^r$, a partition $(I_1,\ldots,I_{r'})$ of $\{1,\ldots,n\}$ and  vectors $\bfc^1,\ldots,\bfc^{r'}\in\{0,1\}^m$ such that 
$$\sum_{i=1}^{r'}\sum_{j\in I_i}\hdist(\bfc^i,\bfa^j)\leq k$$
and the dimension  of the linear space spanned by the vectors  $\bfc^1,\ldots,\bfc^{r'}$ is at most $r$. Note that given a partition $\{I_1,\ldots,I_{r'}\}$ of $\{1,\ldots,n\}$, we cannot select $\bfc^1,\ldots,\bfc^{r'}$ using the majority rule like the case of  \probAtMostClust because of the rank conditions on these vectors. But given $\bfc^1,\ldots,\bfc^{r'}$, one can construct an optimal partition  $\{I_1,\ldots,I_{r'}\}$ with respect to these vectors in the same way as before for \probAtMostClust. Similar to Observation~\ref{obs:eq}, we can restate the task of \probFact.

\begin{observation}\label{obs:eq-fact}
The task of \probFact of binary matrix $\bfA$ with columns $\bfa^1, \dots, \bfa^n$ can equivalently be stated as follows: decide whether there is a positive integer $r'\leq r$  and linearly independent vectors $\bfc^1,\ldots,\bfc^{r'}\in\{0,1\}^m$ over \GF{} such that 
$\sum_{i=1}^n\min\{\hdist(\bfs,\bfa^i)\mid \bfs=\bigoplus_{j\in I}\bfc^j,~I\subseteq \{1,\ldots,r'\}\}\leq k$.
\end{observation} 

Recall that it was proved by Fomin et al.~\cite{FominLMSZ17}  that \probFact is \classFPT when parameterized by $k$ and $r$. To demonstrate that the total dependency on $k+r$ could be relatively small, we observe the following.

\begin{proposition}\label{prop:fact-rk}
\probFact is solvable in time $2^{\Oh(k\log r)}\cdot(nm)^{\Oh(1)}$.
\end{proposition}

\begin{proof}In what follows by rank we  mean the \GFrank of a matrix.  It is more convenient for this algorithm to interpret 
 \probFact  as a matrix editing problem. Given  a matrix $\bfA$ over \GF, a positive integer $r$ and a nonnegative integer $k$, decide whether it is possible to obtain from $\bfA$  a matrix $\bfB$ with  $\rank(\bfB)\leq r$ by editing at most $k$ elements, i.e., by replacing $0$s by $1$s and $1$s by $0$s. We use this to construct a recursive branching algorithm for the problem.

Let $(\bfA =(a_{ij}),r,k)$ be an instance of \probFact. 
 The algorithm for $(\bfA,r,k)$ works as follows.

\begin{itemize}
\item If $\rank(\bfA)\leq r$, then return YES and stop.
\item If $k=0$, then return NO and stop.
\item 
Since the rank of $\bfA$ is more than $r$, there are $r+1$ columns  $I\subseteq\{1,\ldots,m\}$ and $r+1$ rows $J\subseteq\{1,\ldots,n\}$  such that the induced submatrix $\bfA[I,J]$ is of rank $r+1$.  
 We branch into $(r+1)^2$ subproblems: For each $i\in I$ and $j\in J$ we do the following:
\begin{itemize}
\item construct  matrix $\bfA'$ from $\bfA$ by replacing $a_{ij}$ with $a_{ij}\oplus 1$, 
\item call the algorithm for $(\bfA',r,k-1)$ and 
\begin{itemize}
\item if the algorithm returns YES, then return YES and stop.
\end{itemize}
\end{itemize}
\item Return NO and stop.
\end{itemize}

To show the correctness of the algorithm, we observe the following. Let $\bfB$ be an $m\times n$-matrix of rank at most $r$. If $\rank(\bfA[I,J])>r+1$ for some   $I\subseteq\{1,\ldots,m\}$ and $J\subseteq\{1,\ldots,n\}$, then $\|\bfA[I,J]-\bfB[I,J]\|_F^2\geq 1$, i.e, $\bfA[I,J]$ and $\bfB[I,J]$ differ in at least one element. 
To evaluate the running time, notice that we can compute $\rank(\bfA)$ in polynomial time, and if $\rank(\bfA)>r$, then we can find in polynomial time an $(r+1)\times (r+1)$-submatrix of $\bfA$ of rank $r+1$. Then we have $(r+1)^2$ branches in our algorithm. Since we decrease the parameter $k$ in every recursive call, the depth of the recurrence  tree is at most $k$. It implies that the algorithm runs in time $(r+1)^{2k}\cdot(nm)^{\Oh(1)}$.  
\end{proof}

\subsection{Properties of \probPClust}
We will be using the following observation which follows directly from the definition of a $\bfP$-matrix.

\begin{observation}\label{obs:p-matr-numberofrows}
Let $\bfP$ be a binary $p\times q$ matrix. Then every $\bfP$-matrix $\bfB$  
has at most $p$ pairwise-distinct rows and at most $q$ pairwise-distinct columns.
\end{observation} 
  
In our algorithm for \probPClust, we need a subroutine for checking whether a matrix $\bfA$ is a $\bfP$-matrix. For that we employ the following brute-force algorithm.
Let $\bfA$ be an $m\times n$-matrix. Let $\bfa_1,\ldots,\bfa_m$ be the rows of $\bfA$, and let 
$\bfa^1,\ldots,\bfa^n$ be the columns of $\bfA$. 
Let $\mathcal{I}=\{I_1,\ldots,I_s\}$ be the partition of $\{1,\ldots,m\}$ into inclusion-maximal sets of indices such that for every $i\in\{1,\ldots,s\}$  the rows $\bfa_j$ for $j\in I_i$ are equal. Similarly, let 
$\mathcal{J}=\{J_1,\ldots,J_t\}$ be the partition of $\{1,\ldots,n\}$ into inclusion-maximal sets such that for every $i\in\{1,\ldots,t\}$, the columns $\bfa^j$ for $j\in I_i$ are equal. We say that $(\mathcal{I},\mathcal{J})$ is the \emph{block partition} of $\bfA$.

\begin{observation}\label{obs:p-matr-brute}
There is an algorithm which given an $m\times n$-matrix $\bfA= (a_{ij})\in \{0,1\}^{m\times n}$ and a $p\times q$-matrix $\bfP=(p_{ij})\in \{0,1\}^{p\times q}$, runs in time $2^{\Oh(p\log p+q\log q)}\cdot (nm)^{\Oh(1)}$, and decides whether $\bfA$ is a $\bfP$-matrix or not.
\end{observation}

\begin{proof}
 Let $(\mathcal{I}=\{I_1,\ldots,I_s\},\mathcal{J}=\{J_1,\ldots,J_t\})$ 
 be the block partition of $\bfA$ 
 and let $(\mathcal{X}=\{X_1,\ldots,X_{p'}\},\mathcal{Y}=\{Y_1,\ldots,Y_{q'}\})$ 
  be the block partition of $\bfP$.
Observe that $\bfA$ is a $\bfP$-matrix if and only if $s=p'$, $t=q'$ and there are permutations $\alpha$ and $\beta$ of  $\{1,\ldots,p'\}$ and $\{1,\ldots,q'\}$, respectively, such that the following holds for every $i\in\{1,\ldots,p'\}$ and $j\in\{1,\ldots,q'\}$:
\begin{itemize}
\item[i)] $|I_i|\geq |X_{\alpha(i)}|$ and $|J_j|\geq |Y_{\beta(j)}|$,
\item[ii)] $a_{i'j'}=p_{i''j''}$ for $i'\in I_i$, $j\in J_j$, $i''\in X_{\alpha(i)}$ and $j''\in Y_{\beta(i)}$.
\end{itemize}

Thus in order to check whether $\bfA$ is a $\bfP$-matrix, we check whether $s=p'$ and $t=q'$, and if it holds, we consider all possible permutations $\alpha$ and $\beta$ and verify (i) and (ii). Note that 
the block partitions of $\bfA$ and $\bfP$ can be constructed in polynomial time. Since there are $p'!\in 2^{\Oh(p\log p)}$ and $q'!\in 2^{\Oh(q\log q)}$ permutations of $\{1,\ldots,p'\}$ and $\{1,\ldots,q'\}$, respectively, and (i)--(ii) can be verified in polynomial time, we obtain that the algorithm runs in time $2^{\Oh(p\log p+q\log q)}\cdot (nm)^{\Oh(1)}$.
\end{proof}

We conclude the section by showing that \probPClust is \classFPT when parameterized by $k$ and the size of $\bfP$.

\begin{proposition}\label{prop:PClust}
\probPClust can be solved in time $2^{\Oh(k(\log p+\log q)+p\log p+q\log q)}\cdot(nm)^{\Oh(1)}$.
\end{proposition}

\begin{proof}
As with \probFact in Proposition~\ref{prop:fact-rk}, we consider \probPClust as a matrix editing problem. The task now is to obtain from the input matrix $\bfA$ a $\bfP$-matrix by at most $k$ editing operations. We construct a recursive branching algorithm for this.
Let $(\bfA ,\bfP,k)$ be an instance of \probPClust, 
where $\bfA= (a_{ij})\in \{0,1\}^{m\times n}$ and   $\bfP=(p_{ij})\in \{0,1\}^{p\times q}$.
  Then the algorithm  works as follows.

\begin{itemize}
\item Check whether $\bfA$ is a $\bfP$-matrix using Observation~\ref{obs:p-matr-brute}. If it is so, then return YES and stop.
\item If $k=0$, then return NO and stop.
\item Find the block partition $(\mathcal{I},\mathcal{J})$ of $\bfA$. Let $\mathcal{I}=\{I_1,\ldots,I_s\}$ and $\mathcal{J}=\{J_1,\ldots,J_t\}$. 
Set $p'=\min\{s,p+1\}$ and $q'=\min\{t,q+1\}$.
For each $i\in \{1,\ldots,p'\}$ and $j\in \{1,\ldots,q'\}$ do the following:
\begin{itemize}
\item construct  a matrix $\bfA'$ from $\bfA$  by replacing the value of an arbitrary $a_{i'j'}$ for $i'\in I_i$  and $j'\in J_j$ by the opposite value, i.e., set it $a_{i'j'}=1$ if it was 0 and 0 otherwise,
\item call the algorithm recursively for $(\bfA',r,k-1)$, and 
\begin{itemize}
\item if the algorithm returns YES, then return YES and stop.
\end{itemize}
\end{itemize}
\item Return NO and stop.
\end{itemize}

 For the correctness of the algorithm, let us assume that the algorithm did not stop in the first two steps. That is, $\bfA$ is not a $\bfP$-matrix and $k>0$. Consider $I=\bigcup_{i=1}^{p'}I_i$ and $J=\bigcup_{j=1}^{q'}J_j$. Let $\bfB= (b_{ij})\in \{0,1\}^{m\times n}$ be    a $\bfP$-matrix such that $\|\bfA-\bfB\|_F^2\leq k$. 
Observe that $\bfA[I,J]$ and $\bfB[I,J]$ differ in at least one element.  Hence, there is $i\in\{1,\ldots,p'\}$ and $j\in\{1,\ldots,q'\}$ such that $a_{i'j'}\neq b_{i'j'} $ for $i'\in I_i$ and $j'\in J_j$.  Note that for any choice of $i',i''\in I_i$ and $j',j''\in J_j$, the matrices $\bfA'$ and $\bfA''$ obtained from $\bfA$ by   changing the elements  $a_{i'j'}$ and $a_{i''j''}$ respectively, are isomorphic. This implies that $(\bfA,\bfP,k)$ is a \yesinstance of \probPClust if and only if $(\bfA',\bfP,k-1)$ is a \yesinstance for one of the branches of the algorithm.

For the running time evaluation, recall that by Observation~\ref{obs:p-matr-brute}, the first step can be done in time $2^{\Oh(p\log p+q\log q)}\cdot (nm)^{\Oh(1)}$. Then the block partition of $\bfA$ can be constructed in polynomial time and we have at most $(p+1)(q+1)$ recursive calls of the algorithm in the third step. The depth of recursion is at most $k$. Hence, we conclude that  the total running time is   $2^{\Oh(k(\log p+\log q)+p\log p+q\log q)}\cdot(nm)^{\Oh(1)}$.
\end{proof}

\section{\probAtMostClust parameterized by $k$}\label{sec:par-k}
In this section we prove that \probAtMostClust is \classFPT when parameterized by $k$.
That is we prove the following theorem. 

\begin{theorem}\label{thm:fpt-clust-k}
\probAtMostClust is solvable in time $2^{\Oh(k\log k)}\cdot (nm)^{\Oh(1)}$.
\end{theorem}

 The proof of Theorem~\ref{thm:fpt-clust-k} consists of two FPT Turing  reductions. First we define a new auxiliary problem \probClustSelect   
and show how  to reduce  this problem  the \textsc{Consensus Patterns} problem. Then we can use as a black box the algorithm of    Marx~\cite{Marx08} for this problem.
The second reduction is from  \probAtMostClust to \probClustSelect  and is based on  the  color coding technique of Alon, Yuster, and Zwick from~\cite{AlonYZ95}.

\medskip\noindent\textbf{From  \probClustSelect  to  \textsc{Consensus Patterns}.}
In the \probClustSelect\ problem  we are given a regular partition $\{I_1,\ldots,I_p\}$ of columns of matrix $\bfA$.  
Our task is to select from each set $I_i$ exactly one initial cluster such that the total deviation of all the vectors in these clusters from their mean is at most $d$. More formally, 

\defproblema{\probClustSelect}%
{An $m\times n$-matrix $\bfA$ with columns $\bfa^1,\ldots,\bfa^n$, a regular partition $\{I_1,\ldots,I_p\}$ of $\{1,\ldots,n\}$, 
 and a nonnegative integer $d$.}%
{Decide whether there is a set of  initial clusters $J_1,\ldots,J_p$  and a vector $\bfc\in\{0,1\}^m$ such that 
$J_i\subseteq I_i$ for $i\in\{1,\ldots,p\}$ and
$$\sum_{i=1}^{p}\sum_{j\in J_i}\hdist(\bfc,\bfa^j)\leq d.$$
}

If $(\bfA,\{I_1,\ldots,I_p\},d)$ is a \yesinstance of \probClustSelect, then we say that the corresponding sets of initial clusters $\{J_1,\ldots,J_p\}$ and the vector $\bfc$ (or just $\{J_1,\ldots,J_p\}$ as $\bfc$ can be computed by the majority rule from the set of cluster) is a \emph{solution} for the instance.
We show that \probClustSelect is \classFPT when parameterized by $d$. Towards that,  we use the results of Marx~\cite{Marx08} about the \textsc{Consensus Patterns} problem.

\defproblema{\textsc{Consensus Patterns}}%
{A (multi) set of $p$ strings $\{s_1,\ldots,s_p\}$ over an alphabet $\Sigma$, a positive integer $t$ and a nonnegative integer $d$.}%
{Decide whether there is a string $s$ of length $t$ over $\Sigma$, and a length $t$ substring $s_i'$ of $s_i$ for every $i\in\{1,\ldots,p\}$ such that $\sum_{i=1}^p\hdist(s,s_i')\leq d$.
}

Marx proved  that \textsc{Consensus Patterns} can be solved in time $\delta^{\Oh(\delta)}\cdot |\Sigma|^\delta\cdot L^9$ where $\delta=d/p$ and $L$ is the total length of all the strings in the input~\cite{Marx08}. It gives us the following lemma.

\begin{lemma}[\cite{Marx08}]\label{lem:cons}
\textsc{Consensus Patterns} can be solved in time $2^{\Oh(d\log d)}\cdot L^9$, where $L$ is the total length of all the strings in the input if the size of $\Sigma$ is bounded by a constant.
\end{lemma}

Now we are ready to show the following result for \probClustSelect.

\begin{lemma}\label{lem:clustselect}
\probClustSelect can be solved in time $2^{\Oh(d\log d)}\cdot (nm)^{\Oh(1)}$.
\end{lemma}

\begin{proof}
Let $(\bfA,\{I_1,\ldots,I_p\},d)$ be an instance of \probClustSelect. Let $\bfa^1,\ldots,\bfa^n$ be the columns of $\bfA$.
First, we check whether there  are initial clusters $J_1,\ldots,J_p$  and a vector $\bfc=\bfa^i$ for some $i\in\{1,\ldots,n\}$ such that 
$J_j\subseteq I_j$ for $j\in\{1,\ldots,p\}$ and $\sum_{j=1}^{p}\sum_{h\in J_j}\hdist(\bfc,\bfa^h)\leq d$. Towards that 
we consider all possible choices of $\bfc=\bfa^i$ for $i\in\{1,\ldots,n\}$. Suppose that $\bfc$ is given. For every $j\in\{1,\ldots,p\}$, we find an initial cluster $J_j\subseteq I_j$ such that 
$\sum_{h\in J_j}\hdist(\bfc,\bfa^h)$ is minimum. If $\sum_{j=1}^{p}\sum_{h\in J_j}\hdist(\bfc,\bfa^h)\leq d$, then we return the corresponding solution, i.e., the set of initial clusters $\{J_1,\ldots,J_p\}$ and $\bfc$. Otherwise, we discard the choice of $\bfc$.  
It is straightforward to see that  this procedure is correct and can be performed in polynomial time. 
Now on we assume that this is not the case. That  is, if $(\bfA,\{I_1,\ldots,I_p\},d)$ is a \yesinstance, then $\bfc\neq \bfa^i$ for any solution. 
In particular, it means that for every solution $(\{J_1,\ldots,J_p\},\bfc)$, $\hdist(\bfc,\bfa^j)\geq 1$ for $j\in J_1\cup\ldots\cup J_p$. If $p>d$, we obtain that   $(\bfA,\{I_1,\ldots,I_p\},d)$ is a no-instance. In this case we return the answer and stop. Hence, from now we  assume that $p\leq d$. Moreover, observe that  $|J_1|+\ldots+|J_p|\leq d$ for any solution  $(\{J_1,\ldots,J_p\},\bfc)$.  

We consider all $\bfP$-tuples of positive integers $(\ell_1,\ldots,\ell_p)$ such that $\ell_1+\ldots+\ell_p\leq d$ and for each $\bfP$-tuple check whether there is a solution $(\{J_1,\ldots,J_p\},\bfc)$ with $|J_i|=\ell_i$ for $i\in\{1,\ldots,p\}$. Note that there are at most $2^{d+p} \leq 4^d$ such $\bfP$-tuples. If we find a solution for one of the $\bfP$-tuples, we return it and stop. If we have no solution for any $\bfP$-tuple, we conclude that we have a no-instance of the problem.

Assume that we are given a $\bfP$-tuple $(\ell_1,\ldots,\ell_p)$. 
If there is $i\in\{1,\ldots,p\}$ such that there is no initial cluster $J_i\in I_i$ with $|J_i|=\ell_i$, then we discard the current choice of the $\bfP$-tuple. Otherwise, we reduce the instance of the problem using the following rule: if there is $i\in\{1,\ldots,p\}$ and an initial cluster $J\subseteq I_i$ such that $|I|\neq\ell_i$, then delete columns $\bfa^h$ for $h\in J$ from the matrix and set $I_i=I_i\setminus J$. By this rule, we can assume that each $I_i$ contains only initial clusters of size $\ell_i$. Let $I_i=\{J_1^i,\ldots,J_{q_i}^i\}$ where $J_1^i,\ldots,J_{q_i}^i$ are initial clusters for $i\in\{1,\ldots,p\}$.

We reduce the problem of checking the existence of a solution $(\{J_1,\ldots,J_p\},\bfc)$ with $|J_i|=\ell_i$ for $i\in\{1,\ldots,p\}$ to the \textsc{Consensus Patterns} problem.  
Towards that, we first define the alphabet $\Sigma=\{0,1,a,b\}$ and strings  
\begin{eqnarray*}
 \overline{a}=\underbrace{a\ldots a}_{m+d} ,  &  \, \overline{b}=\underbrace{b\ldots b}_{m+d}, \,  & \text{and } \overline{0}=\underbrace{0\ldots 0}_{d}. 
\end{eqnarray*}
Then $x=\overline{a}\overline{b}\cdots\overline{a}\overline{b}$ is defined to be the string obtained by the alternating concatenation of $d+1$ copies of $\overline{a}$ and $d+1$ copies of $\overline{b}$.
 Now we construct $\ell=\ell_1+\ldots+\ell_p$ strings $s_i^j$ for $i\in\{1,\ldots,p\}$ and $j\in\{1,\ldots,\ell_i\}$. 
For each $i\in\{1,\ldots,p\}$, we do the following. 
\begin{itemize}
\item For every $q\in\{1,\ldots,q_i\}$, select a column $\bfa^{h_{i,q}}$ for $h_{i,q}\in J_q^i$ and, by slightly abusing the notation, consider it to be a $(0,1)$-string. 
\item Then for every $j\in \{1,\ldots,\ell_i\}$, set $s_i^j=x\bfa^{h_{i,1}}\overline{0}x\ldots x\bfa^{h_{i,q_i}}\overline{0}x$. 
\end{itemize} 
Observe that the strings $s_i^j$ for $j\in\{1,\ldots,\ell_i\}$ are the same. We denote by $S=\{s_i^j\mid 1\leq i\leq p,~1\leq j\leq \ell_i\}$ the collection (multiset) of all constructed strings.
Finally, we define $t=(m+d)(2d+3)$. Then output $(S,\Sigma,t,d)$ as the instance 
of \textsc{Consensus Patterns}. 
Now we prove the correctness of the reduction.

\begin{claim}
\label{claim:cstr}
The instance $(A,\{I_1,\ldots,I_p\},d)$ of \probClustSelect has a solution $(\{J_1,\ldots,J_p\},\bfc)$ with $|J_i|=\ell_i$ for $i\in\{1,\ldots,p\}$ if and only if $(S,\Sigma,t,d)$ is a \yesinstance of \textsc{Consensus Patterns}.
\end{claim}

\begin{proof}[Proof of Claim~\ref{claim:cstr}]
Suppose that the instance $(A,\{I_1,\ldots,I_p\},d)$ of \probClustSelect has a solution $(\{J_1,\ldots,J_p\},\bfc)$ with $|J_i|=\ell_i$ for $i\in\{1,\ldots,p\}$.
For every $i\in\{1,\ldots,p\}$ and $j\in\{1,\ldots,\ell_i\}$, we select the substring $\hat{s}_i^j=x\bfa^{h_{i,j}}\overline{0}x$ where $h_{i,j}\in J_i$. By the definition, $|\hat{s}_i^j|=2|x|+m+d=t$. We set $s=x\bfc\overline{0}x$ considering the vector $\bfc$ being a $(0,1)$-string. Clearly, $|s|=t$.
We have that
$$\sum_{i=1}^p\sum_{j=1}^{\ell_i}\hdist(s,\hat{s}_i^j)=\sum_{i=1}^p\ell_i\hdist(\bfc,\bfa^{h_{i,j}})= \sum_{i=1}^p\sum_{h\in J_i}\hdist(\bfc,\bfa^h)\leq d.$$
Therefore, $(S,\Sigma,t,d)$ is a \yesinstance of \textsc{Consensus Patterns}.

Now we prove the reverse direction. 
Assume that $(S,\Sigma,t,d)$ is a \yesinstance of \textsc{Consensus Patterns}. Let $\hat{s}_i^j$ be a substring of $s_i^j$ of length $t$ for $i\in\{1,\ldots,p\}$ and $j\in\{1,\ldots,\ell_i\}$, and let $s$ be a string of length $t$ over $\Sigma$ such that 
$\sum_{i=1}^p\sum_{j=1}^{\ell_i}\hdist(s,\hat{s}_i^j)\leq d$.
We first show that there is a positive integer $\alpha\leq t-m+1$ such that for every $i\in \{1,\ldots,p\}$ and $j\in\{1,\ldots,\ell_i\}$, $\hat{s}_i^j[\alpha..\alpha+m-1]=\bfa^{h_{i,j}}$ for some $h_{i,j}\in I_i$.

Consider the substring $\hat{s}_1^1$. Since $|\hat{s}_1^1|=t$, by the definition of the string $s_1^1$, we have that there is a positive integer $\beta\leq t-|x|+1=t-2(m+d)(d+1)+1$ such that $\hat{s}_1^1[\beta..\beta+\vert x \vert-1]=\hat{s}_1^1[\beta..\beta+2(m+d)(d+1)-1]=x$. Let  $i\in \{1,\ldots,p\}$ and $j\in\{1,\ldots,\ell_i\}$. Suppose that $\hat{s}_i^j[\beta..\beta+2(m+d)(d+1)-1]\neq x$. Recall that $x$ contains $2(d+1)$ alternating copies of $\overline{a}$ and $\overline{b}$ and $|\overline{a}|=|\overline{b}|=m+d$.
Because $\hdist(\hat{s}_1^1,\hat{s}_i^j)\leq \hdist(\hat{s}_1^1,{s})+\hdist({s},\hat{s}_i^j) \leq d$ (by the triangular inequality) and by the construction of the strings of $S$, we have that that either
\begin{itemize} 
\item[i)] there is $\gamma=\beta+2(m+d)h$ for some nonnegative integer $h\leq d$ such that   $\hat{s}_i^j[\beta..\gamma-1]=x[\beta..\gamma-1]$ and  $\hat{s}_i^j[\gamma..\gamma+m-1]=\bfa^{h_{i,j}}$ for some $h_{i,j}\in I_i$, or
\item[ii)] there is $\gamma=\beta+2(m+d)h$ for some integer $1\leq h\leq d$ such that   $\hat{s}_i^j[\gamma..\beta+2(m+d)(d+1)-1]=x[\gamma..\beta+2(m+d)(d+1)-1]$ and  $\hat{s}_i^j[\gamma-(m+d)..\gamma-d+1]=\bfa^{h_{i,j}}$ for some $h_{i,j}\in I_i$.
\end{itemize}

We would like to mention that in the above two cases, in one of them $x[\beta..\gamma-1]$ or  $x[\gamma..\beta+2(m+d)(d+1)-1]$ may not be well defined. But at least in one of them it will be well defined.  
These cases are symmetric and without loss of generality we can consider only the case (i). We have that  $\hat{s}_i^j[\gamma..\gamma+(m+d)-1]$ and $\hat{s}_1^1[\gamma..\gamma+(m+d)-1]$ differ in all the symbols, because all the symbols of $\hat{s}_i^j[\gamma..\gamma+(m+d)-1]$ are $0$ or $1$ and $\hat{s}_1^1[\gamma..\gamma+(m+d)-1]=\overline{a}$.
Since $m+d>d$, it contradicts the property that $\hdist(\hat{s}_1^1,\hat{s}_i^j)\leq d$.
So we have that $\hat{s}_i^j[\beta..\beta+2(m+d)(d+1)-1]=x$ for every $i\in \{1,\ldots,p\}$ and $j\in\{1,\ldots,\ell_i\}$. 

If $\beta>m+d$, then we set $\alpha=\beta-(m+d)$. 
Notice that that for every $i\in \{1,\ldots,p\}$ and $j\in\{1,\ldots,\ell_i\}$, $\hat{s}_i^j[\alpha..\alpha+m-1]=\bfa^{h_{i,j}}$ for some $h_{i,j}\in I_i$.  Suppose $\beta\leq m+d$. Then we set $\alpha=\beta+2(m+d)(d+1)$. 
Because $t=(m+d)(2d+3)$, 
it holds that $\alpha\leq t-m+1$ and for every $i\in \{1,\ldots,p\}$ and $j\in\{1,\ldots,\ell_i\}$, $\hat{s}_i^j[\alpha..\alpha+m-1]=\bfa^{h_{i,j}}$ for some $h_{i,j}\in I_i$.
Now 
consider $c=s[\alpha..\alpha+m-1]$. Because for every $i\in \{1,\ldots,p\}$ and $j\in\{1,\ldots,\ell_i\}$, $\hat{s}_i^j[\alpha..\alpha+m-1]=\bfa^{h_{i,j}}$ for some $h_{i,j}\in I_i$, we can assume that $\bfc$ is a $(0,1)$-string. We consider it as a vector of $\{0,1\}^m$. 

Let $i\in\{1,\ldots,p\}$. We consider the columns $\bfa^{h_{i,j}}$ for $j\in
\{1,\ldots,\ell_i\}$ and find among them the column $\bfa^{h_i}$ such that $\hdist(\bfc,\bfa^{h_i})$ is minimum. Let $J_{r_i}^i\subseteq I_i$ for $r_i\in\{1,\ldots,q_i\}$ be an initial cluster that contains $h_i$.
We show that $(\{J_{r_1}^1,\ldots,J_{r_p}^p\},\bfc)$ is a solution for the instance $(\bfA,\{I_1,\ldots,I_p\},d)$ of \probClustSelect. To see it, it is sufficient to observe that the following inequality  holds.
\begin{align*}
\sum_{i=1}^p\sum_{h\in J_{r_i}^i}\hdist(\bfc,\bfa^h)=&\sum_{i=1}^p\ell_i\hdist(\bfc,\bfa^{h_i})\leq\sum_{i=1}^p\sum_{j=1}^{\ell_i}\hdist(\bfc,\bfa^{h_{i,j}})\leq \sum_{i=1}^p\sum_{j=1}^{\ell_i}\hdist(s,\hat{s}_i^j)\leq d.
\end{align*}
This concludes the proof of the claim.
\end{proof}

Using Claim~\ref{claim:cstr} and Lemma~\ref{lem:cons}, we solve \textsc{Consensus Patterns} for $(S,\Sigma,t,d)$. 
This completes the description of our algorithm and its correctness proof. To evaluate the running time, recall first that we check in polynomial time whether we have a solution  
with $\bfc$ coinciding 
with a column of $\bfA$. If we fail to find such a solution, then we consider at most $4^d$ $\bfP$-tuples $(\ell_1,\ldots,\ell_p)$. Then for each $\bfP$-tuple, we either discard it immediately or construct in polynomial time the corresponding instance of \textsc{Consensus Patterns}, which  is solved in time $2^{\Oh(d\log d)}\cdot (nm)^{\Oh(1)}$ by Lemma~\ref{lem:cons}. Hence, the total running time is  $2^{\Oh(d\log d)}\cdot (nm)^{\Oh(1)}$.
\end{proof}

Let us  note that we are using Lemma~\ref{lem:cons} as a black box in our algorithm for \probClustSelect. By adapting the algorithm of Marx~\cite{Marx08} for \textsc{Consensus Patterns} to solve \probClustSelect it is possible to improve the polynomial factor in the running time but this would demand repeating and rewriting  various parts of~\cite{Marx08}.

\medskip\noindent\textbf{From  \probAtMostClust to  \textsc{Consensus Patterns}.}
Now we prove the main result of the section.

\begin{proof}[Proof of Theorem~\ref{thm:fpt-clust-k}]
Our algorithm for \probAtMostClust uses the \emph{color coding} technique introduced by Alon, Yuster and Zwick in~\cite{AlonYZ95} (see also~\cite{CyganFKLMPPS15} for the introduction to this technique). In the end we obtain a deterministic algorithm but it is more convenient for us to describe a randomized true-biased Monte-Carlo algorithm and then explain how it could be derandomized.

Let $(\bfA,r,k)$ be a \yesinstance of \probAtMostClust where $\bfA=(\bfa^1,\ldots,\bfa^n)$. Then by Lemma~\ref{lem:init}, there is a regular solution $\{I_1,\ldots,I_{r'}\}$ for this instance. Let $\bfc^1,\ldots,\bfc^{r'}$ be the corresponding means of the clusters. Recall that regularity means that for any initial cluster $I$, there is a cluster in the solution that contains it. We say that a cluster $I_i$ of the solution is \emph{simple} if it contain exactly one initial cluster and $I_i$ is \emph{composite} otherwise.
Let $I_i$ be a composite cluster of $\{I_1,\ldots,I_{r'}\}$ that contains $h\geq 2$ initial clusters. Then $\sum_{j\in I_i}(\bfc^i,\bfa^j)\geq h-1$. This observation immediately implies that a regular solution contains at most $k$ composite clusters and the remaining clusters are simple. Moreover, the total number of initial clusters in the composite clusters is at most $2k$.
Note also that if $I_i$ is a simple cluster then $\bfc^i=\bfa^h$ for arbitrary $h\in I_i$, because   $\sum_{j\in I_i}(\bfc^i,\bfa^j)=0$, That is, simple clusters do not contribute to the total cost of the solution. 

Let $(\bfA,r,k)$ be an instance of \probAtMostClust where $\bfA=(\bfa^1,\ldots,\bfa^n)$. We construct the set $\mathcal{I}$ of initial clusters for the matrix $\bfA$. Let $s=|\mathcal{I}|$. The above observations imply  that finding a solution for \probAtMostClust is equivalent to finding a set $\mathcal{I}'\subseteq \mathcal{I}$ of size at most $2k$ such that $\mathcal{I}'$ can be partitioned  into at most $r-s+|\mathcal{I}'|$ composite clusters. More precisely, we are looking for  $\mathcal{I}'\subseteq \mathcal{I}$ of size at most $2k$ such that there is a partition $\{P_1,\ldots,P_t\}$  of 
$\mathcal{I}'$ with $t\leq r-s+|\mathcal{I}'|$ and vectors $\bfs^1,\dots,\bfs^t\in \{0,1\}^m$ with the property that
$$\sum_{i=1}^t\sum_{I\in P_i}\sum_{j\in I}\hdist(\bfs^i,\bfa^h)\leq k.$$

If $s\leq r$ the $(\bfA,r,k)$ is a trivial \yesinstance of the problem with $\mathcal{I}$ being a solution. If $r+k< s$, then $(\bfA,r,k)$ is a trivial no-instance. Hence, we assume from now that $r< s\leq r+k$.
We color the elements of $\mathcal{I}$ independently and uniformly at random by $2k$ colors $1,\ldots,2k$. Observe that if $(\bfA,r,k)$ is a \yesinstance, then at most $2k$ initial clusters in a solution that are included in composite clusters are colored by distinct colors with the probability at least $\frac{(2k)!}{(2k)^{2k}}\geq e^{-2k}$.  We say that a solution  $\{I_1,\ldots,I_{r'}\}$ for $(\bfA,r,k)$ is a \emph{colorful} solution if all initial clusters that are included in composite clusters of  $\{I_1,\ldots,I_{r'}\}$ are colored by distinct colors. 
We construct an algorithm for finding a colorful solution (if it exists).

Denote by $\mathcal{I}_1,\ldots,\mathcal{I}_{2k}$ the sets of color classes of initial clusters, i.e., the sets of initial clusters that are colored by $1,\ldots,{2k}$, respectively. Note that some sets could be empty. We consider all possible  partitions $\mathcal{P}=\{P_1,\ldots,P_t\}$ of nonempty subsets of $\{1,\ldots,2k\}$ such that each set of $\mathcal{P}$ contains at least two elements. . 
Notice that if $(\bfA,r,k)$ has a colorful solution $\{I_1,\ldots,I_{r'}\}$, then there is $\mathcal{P}=\{P_1,\ldots,P_t\}$ such that a cluster $I_i$ of the solution is composite cluster containing initial clusters colored by a set of colors $X_i$ if and only if there is $X_i\in\mathcal{P}$. Since we consider all possible $\mathcal{P}$, if $(\bfA,r,k)$ has a colorful solution, we will find $\mathcal{P}$ satisfying this condition.
Assume that $\mathcal{P}=\{P_1,\ldots,P_t\}$ is given. 
If $s-|P_1|-\ldots-|P_t|+t>r$, we discard the current choice of $\mathcal{P}$. Assume from now that this is not the case.

For each $i\in\{1,\ldots,t\}$, we do the following. Let $P_i=\{i_1,\ldots,i_p\}\subseteq \{1,\ldots,2k\}$. Let 
$J_j^i=\bigcup_{I\in \mathcal{I}_{i_j}} I$ and $J^i=J_1^i\cup\ldots\cup J_p^i$. Denote by $\bfA_i$ the submatrix of $\bfA$ containing the columns $\bfa^h$ with $h\in J^i$.  We use Lemma~\ref{lem:clustselect} to find the minimum nonnegative integer $d_i\leq k$ such that $(\bfA_i,\{J_1^i,\ldots,J_p^i\},d_i)$ is a \yesinstance of \probClustSelect. If such a value of $d_i$ does not exist, we discard the current choice of $\mathcal{P}$.  Otherwise, we find the corresponding solution $(\{L_1^i,\ldots,L_p^i\},\bfs^i)$ of \probClustSelect. Let $L^i=L_1^i\cup\ldots\cup L_p^i$

If we computed $d_i$ and constructed $L^i$ for all $i\in\{1,\ldots,t\}$, we check whether $d_1+\ldots+d_t\leq k$. If it holds, we return the colorful solution with the composite clusters $L^1,\ldots,L^t$ whose means are $\bfs^1,\ldots,\bfs^t$ respectively and the remaining clusters are simple. Otherwise, we discard the choice of $\mathcal{P}$.
If for one of the choices of $\mathcal{P}$ we find a colorful solution, we return it and stop. If we fails to find a solution for all possible choices of $\mathcal{P}$, we return the answer NO and stop.

If the described algorithm produces a solution, then it is straightforward to verify that this is a colorful solution to $(\bfA,r,k)$ recalling that simple clusters do not contribute to the total cost of the solution. In the other direction, if $(\bfA,r,k)$ has  a colorful solution $\{I_1,\ldots,I_{r'}\}$, then  there is $\mathcal{P}=\{P_1,\ldots,P_t\}$ such that cluster $I_i$ of the solution is  a composite cluster containing initial clusters colored by a set of colors $X_i$ if and only if there is $X_i\in\mathcal{P}$. Let $L_1,\ldots,L_t$ be the composite clusters of the solution that correspond to $P_1,\ldots,P_t$, respectively and denote by $\bfs^1,\ldots,\bfs^t$ their means. Let $d_i=\sum_{h\in L_i}\hdist(\bfs^i,\bfa^h)$ for $i\in\{1,\ldots,t\}$. It immediately follows that for each $i\in\{1,\ldots,t\}$, it holds that if $P_i=\{i_1,\ldots,i_p\}$, then the constructed  instance $(\bfA_i,\{J_1^i,\ldots,J_p^i\},d_i)$ of \probClustSelect is a \yesinstance. Hence, the algorithm returns 
 a colorful solution.

To evaluate the running time, recall that we consider  $2^{\Oh(k\log k)}$ partitions $\mathcal{P}=\{P_1,\ldots,P_t\}$ of nonempty subsets of $\{1,\ldots,2k\}$. Then for each $\mathcal{P}$, we construct in polynomial time at most $k|\mathcal{P}|$ instances of \probClustSelect. These instances are solved in time $2^{\Oh(k\log k)}\cdot(nm)^{\Oh(1)}$ by Lemma~\ref{lem:clustselect}. We conclude that the total running time of the algorithm that checks the existence of a colorful solution is $2^{\Oh(k\log k)}\cdot(nm)^{\Oh(1)}$.

Clearly, if for a random coloring of $\mathcal{I}$, there is a colorful solution to $(\bfA,r,k)$, then $(\bfA,r,k)$ is a \yesinstance.  We consider $N=\lceil e^{2k}\rceil$ random colorings of $\mathcal{I}$ and for each coloring, we check the existence of a colorful solution. If we find such a solution, we return it and stop. Otherwise, if we failed to find a solution for all colorings, we return the answer NO. 
Recall that if $(\bfA,r,k)$ is a \yesinstance with a solution $\{I_1,\ldots,I_{r'}\}$, then the initial clusters that are included in the composite clusters of the solution are colored  by distinct colors with the probability at least $\frac{(2k)!}{(2k)^{2k}}\geq e^{-2k}$.  Hence, the probability that a \yesinstance has no colorful solution is at most $(1-e^{-2k})$ and, therefore, the probability that 
a \yesinstance has no colorful solution for $N\geq e^{2k}$ random colorings is at most $(1-e^{-2k})^{e^{2k}}\leq e^{-1}$.  We conclude that our randomized algorithm returns a false negative answer with probability at most $e^{-1}<1$. The total running time  of the algorithm is $N\cdot 2^{\Oh(k\log k)}\cdot (nm)^{\Oh(1)}$, that is, $2^{\Oh(k\log k)}\cdot (nm)^{\Oh(1)}$.

By the standard derandomization technique using perfect hash families, see \cite{AlonYZ95,NaorSS95}, our algorithm can be derandomized.  Thus, we conclude that 
\probAtMostClust is solvable in the deterministic time $2^{\Oh(k\log k)}\cdot (nm)^{\Oh(1)}$.
\end{proof}

\section{Kernelization for \probAtMostClust}\label{sec:kernel}
In this section we show that \probAtMostClust admits a polynomial kernel when parameterized by $r$ and $k$. Then we complement this result and Theorem~\ref{thm:fpt-clust-k} by proving that it is unlikely that the problem has a polynomial kernel when parameterized by $k$ only.
Let us start from the definition of a kernel, here we follow \cite{CyganFKLMPPS15}. 

Roughly speaking, kernelization is a {\em{preprocessing algorithm}} that consecutively applies various data reduction rules in order to shrink the instance size as much as possible. Thus, such a preprocessing algorithm takes as input an instance $(I,k)\in\Sigma^{*}\times\mathbb{N}$ of $Q$, works in polynomial time, and returns an equivalent instance $(I',k')$ of $Q$. The quality of kernelization algorithm  $\mathcal{A}$ is measured by the size of the output. More precisely,  the {\em{output size}} of a preprocessing algorithm $\mathcal{A}$ is a function $\textrm{size}_{\mathcal{A}}\colon \mathbb{N}\to \mathbb{N}\cup \{\infty\}$ defined as follows:
$$\textrm{size}_{\mathcal{A}}(k) = \sup \{ |I'|+k'\ \colon\ (I',k')=\mathcal{A}(I,k),\ I\in \Sigma^{*}\}.$$

\begin{definition} 
A {\em{kernelization algorithm}}, or simply a {\em{kernel}}, for a parameterized problem $Q$ is an algorithm $\mathcal{A}$ that, given an instance $(I,k)$ of $Q$, works in polynomial time and returns an equivalent instance $(I',k')$ of $Q$. Moreover,    $\textrm{size}_{\mathcal{A}}(k)\leq g(k)$ for some computable function $g\colon \mathbb{N}\to \mathbb{N}$.
\end{definition}
If the upper bound $g(\cdot)$ is a polynomial   function of the parameter, then we say that $Q$ admits a {\emph{polynomial   kernel}}.

\subsection{Polynomial kernel with parameter $k+r$.}\label{subsec:polykr}

\begin{theorem}\label{thm:kernel}
\probAtMostClust parameterized by $r$ and $k$ has a kernel of size  $\Oh(k^3(k+r)^2)$. Moreover, the kernelization  algorithm in polynomial time either solves the problem or outputs an instance of \probAtMostClust with the matrix that has at most $k+r$ pairwise distinct columns and $\Oh(k^2(k+r))$ pairwise distinct rows.
\end{theorem}

\begin{proof}
Let $(\bfA,r,k)$ be an instance of \probAtMostClust. Let $\bfa^1,\ldots,\bfa^n$ be the columns of $\bfA=(a_{ij})\in \{0,1\}^{m\times n}$.
We apply  the following sequence  of reduction rules.

\begin{reduction}\label{rule:trivial}
If $\bfA$ has at most $r$ pairwise distinct columns then output a trivial \yesinstance  and stop. If  $\bfA$ has at least $k+r+1$ pairwise distinct columns then output a trivial \noinstance and stop.
\end{reduction}

Let us remind that an initial cluster is an inclusion maximal set of equal columns of the input matrix. 
To show that the rule is sound, observe first that if  $\bfA$ has at most $r$ pairwise distinct columns, then the initial clusters form a solution. Therefore, $(\bfA,r,k)$ is a \yesinstance. 
Suppose that $(\bfA,r,k)$ is a \yesinstance of \probAtMostClust. By Observation~\ref{obs:eq}, there is a  set of means $\{\bfc^1,\ldots,\bfc^{r'}\}$ for some $r'\leq r$ such that $\sum_{i\in 1}^n\min\{\hdist(\bfc^j,\bfa^i)\mid 1\leq j\leq r'\}\leq k$. It immediately implies that $\bfA$ has at most $k$ columns that are distinct from $\bfc^1,\ldots,\bfc^{r'}$. Therefore, $\bfA$ has at most $r+k$ distinct columns. 

Assume that Reduction Rule~\ref{rule:trivial} is not applicable on the input instance.
Then we exhaustively apply the following rule.

\begin{reduction}\label{rule:initial-red}
If $\bfA$ has an initial cluster $I\subseteq\{1,\ldots,n\}$ with $|I|>k+1$, then delete a column $\bfa^i$ for $i\in I$.
\end{reduction}
\begin{claim}\label{claim:ruledelte}Reduction Rule~\ref{rule:initial-red} is sound.
\end{claim}
\begin{proof}[Proof of Claim~\ref{claim:ruledelte}]
To show that the rule is sound, assume that $\bfA'$ is obtained from $\bfA$ by the application of Reduction Rule~\ref{rule:initial-red} to the initial cluster $I$ and let $\bfa^i$ be the deleted column. We use the notation $\bfa^1,\ldots,\bfa^{i-1},\bfa^{i+1},\ldots\bfa^n$ for the columns of $\bfA'$.

Suppose that  $(\bfA,r,k)$ is a \yesinstance. Let $\{I_1,\ldots,I_{r'}\}$ be a solution with the means $\bfc^1,\ldots,\bfc^{r'}$. For $j\in\{1,\ldots,r'\}$, let $J_j=I_j\setminus \{i\}$. 
We have that 
$$k\geq \sum_{j=1}^{r'}\sum_{h\in I_j}\hdist(\bfc^j,\bfa^h)\geq \sum_{j=1}^{r'}\sum_{h\in J_j}\hdist(\bfc^j,\bfa^h)$$
and, therefore, $\{J_1,\ldots,J_{r'}\}$ is a solution for  $(\bfA',r,k)$. Therefore, if  $(\bfA,r,k)$ is a \yesinstance, then 
 $(\bfA',r,k)$ is a \yesinstance. 

Now assume that  $(\bfA',r,k)$ is a \yesinstance. Then by Lemma~\ref{lem:init},  $(\bfA',r,k)$  admits a regular solution  $\{J_1,\ldots,J_{r'}\}$  and we have that $I'\subseteq J_s$ for some $s\in\{1,\ldots,r'\}$. Denote by $\bfc^1,\ldots,\bfc^{r'}$ the means of $J_1,\ldots,J_{r'}$ obtained by the majority rule. We have that 
$$k\geq \sum_{j=1}^{r'}\sum_{h\in J_j}\hdist(\bfc^j,\bfa^h)\geq \sum_{h\in J_s}\hdist(\bfc^s,\bfa^h)\geq \sum_{h\in I'}\hdist(\bfc^s,\bfa^h).$$
Since $|I'|\geq k+1$, we have that $\bfc^s=\bfa^h$ for $h\in I'$. Hence, $\bfc^s=\bfa^i$. Let 
$I_j=J_j$ for $j\in\{1,\ldots,r'\}$, $j\neq s$, and let $I_s=J_s\cup\{i\}$. Because $\bfc^s=\bfa^i$, we obtain that 
$$
 \sum_{j=1}^{r'}\sum_{h\in I_j}\hdist(\bfc^j,\bfa^h)= \sum_{j=1}^{r'}\sum_{h\in J_j}\hdist(\bfc^j,\bfa^h)\leq k
$$
and $\{I_1,\ldots,I_{r'}\}$ is a solution for  $(\bfA,r,k)$. That is,  if $(\bfA',r,k)$ is a \yesinstance, then $(\bfA,r,k)$ is also a \yesinstance. This completes the soundness proof.
\end{proof}

To simplify notations, assume that $\bfA$ with the columns $\bfa^1,\ldots,\bfa^n$ 
is the instance of \probAtMostClust obtained by the exhaustive application of 
Reduction Rule~\ref{rule:initial-red}. Note that by Reduction Rules~\ref{rule:trivial} and \ref{rule:initial-red},  $\bfA$ has at most $k+r$ pairwise distinct columns and $n\leq (k+1)(k+r)$.
This means that we have that the number of columns is bounded by a polynomial of $k$ and $r$. However,  the number of rows still could be large. 
Respectively, our aim now is to construct an equivalent instance with the bounded number of rows.

 We greedily construct the partition $\mathcal{S}=\{S_1,\ldots,S_s\}$ of $\{1,\ldots,n\}$ using the following algorithm. Let $i\geq 1$ be an integer and suppose that 
the sets $S_0,\ldots,S_{i-1}$ are already constructed assuming that $S_0=\emptyset$. Let $I=\{1,\ldots,n\}\setminus\bigcup_{j=0}^{i-1}S_j$. If $I\neq \emptyset$, we construct $S_i$:
\begin{itemize}
\item set $S_i=\{s\}$ for arbitrary $s\in I$ and set $I=I\setminus\{s\}$,
\item while there is $j\in I$ such that $\hdist(\bfa^j,\bfa^h)\leq k$ for some $h\in S_i$, then set $S_i=S_i\cup\{j\}$ and  set $I=I\setminus\{j\}$.
\end{itemize}

The crucial property of the partition $\mathcal{S}$ is  that   every cluster of a solution solution is entirely in some of part of the partition.
This way $\mathcal{S}$   separates the clustering problem into subproblems. More precisely, 

\begin{claim}\label{claim:sep}
Let $\{I_1,\ldots,I_{r'}\}$ be a solution for  $(\bfA,r,k)$. Then for every $i\in\{1,\ldots,r'\}$ there is $j\in\{1,\ldots,s\}$ such that $I_i\subseteq S_j$. 
\end{claim}

\begin{proof}[Proof of Claim~\ref{claim:sep}]
Let $\bfc^1,\ldots,\bfc^{r'}$ be the means of $I_1,\ldots,I_{r'}$ obtained by the majority rule. For the sake of contraction, assume that there is a cluster $I_i$ such that there are $p,q\in I_i$ with $p$ and $q$ in distinct sets of the partition $\{S_1,\ldots,S_s\}$. Then $\hdist(\bfa^p,\bfa^q)>k$. Therefore,
$$ \sum_{j=1}^{r'}\sum_{h\in I_j}\hdist(\bfc^j,\bfa^h)\geq \sum_{h\in I_i}\hdist(\bfc^i,\bfa^h)\geq \hdist(\bfc^i,\bfa^p)+\hdist(\bfc^i,\bfa^q)\geq \hdist(\bfa^p,\bfa^q)>k
$$
contradicting that $\{I_1,\ldots,I_{r'}\}$ is a solution.
\end{proof}

We say that a row of a binary matrix is \emph{uniform} if all its elements are equal. Thus a uniform row consists entirely from $0$s or from $1$s. Otherwise,  a row is \emph{nonuniform}. We show that the submatrices of $\bfA$ composed by the columns with indices from  the same part of partition  of $\mathcal{S}$,  have a bounded number of nonuniform rows.

\begin{claim}\label{claim:uniform}
For every $i\in\{1,\dots,s\}$, the matrix $\bfA[\{1,\ldots,m\},S_i]$ has at most $(|S_i|-1)k$ nonuniform columns. 
\end{claim}

\begin{proof}[Proof of Claim~\ref{claim:uniform}]
Let $i\in\{1,\dots,s\}$ and $\ell=|S_i|$. Recall that $S_i$ is constructed greedily by adding the index of a column that is at distance at most $k$ from some columns whose index is already included in $S_i$. Denote respectively by $I_1,\ldots,I_\ell$ the sets constructed on each iteration. 

For every $j\in\{1,\ldots,\ell\}$, we show inductively that  $\bfA[\{1,\ldots,m\},I_j]$ has at most $(j-1)k$ nonuniform columns. The claim is trivial for $j=1$. Let $I_2=\{p,q\}$ for some distinct $p,q\in S_i$, then because $\hdist(\bfa^p,\bfa^q)\leq k$, we have that $\bfa^p$ and $\bfa^q$ differ in at most $k$ positions and, therefore,  $\bfA[\{1,\ldots,m\},I_2]$ has at most $k$ nonuniform rows. 

Let $j\geq 3$. Assume that $p\in I_j\setminus I_{j-1}$. By the inductive assumption, $\bfA[\{1,\ldots,m\},I_{j-1}]$ has at least $m-(j-2)k$ uniform rows. Denote by $J\subseteq\{1,\ldots,m\}$ the set of indices of uniform rows of $\bfA[\{1,\ldots,m\},I_{j-1}]$. Since $\hdist(\bfa^p,\bfa^q)\leq k$ for some $q\in I_{i-1}$, there are at most $k$ positions where  $\bfa^p$ and $\bfa^q$ are distinct. In particular, there are at most $k$ indices of $J$ for which the corresponding elements of  $\bfa^p$ and $\bfa^q$ are distinct. This immediately implies that $\bfA[\{1,\ldots,m\},I_j]$ has at least $|J|-k\geq m-(j-2)k-k=m-(j-1)k$ uniform rows. Hence, $\bfA[\{1,\ldots,m\},I_j]$ has at most $(j-1)k$ nonuniform rows.
\end{proof}

Now we proceed with the kernelization algorithm. The next rule is used to deal with trivial cases. 

\begin{reduction}\label{rule:bigS}
If $s>r$, then  output a trivial \noinstance  and stop.
\end{reduction}

The soundness of the rule is immediately implied by  Claim~\ref{claim:sep}.
From now we  assume that $s\leq r$. 

Observe that if all the rows of the matrix $\bfA[\{1,\ldots,m\},S_i]$ are uniform, then $S_i$ is an initial cluster. Since  after the application of Reduction Rule~\ref{rule:trivial} the number of initial clusters is at least $r+1$, we have that for some $i\in\{1,\ldots,s\}$, matrix $\bfA[\{1,\ldots,m\},S_i]$ contains nonuniform rows. Let $\ell_i$ be the number of nonuniform rows in $\bfA[\{1,\ldots,m\},S_i]$ for $i\in\{1,\ldots,s\}$, and let $\ell=\max_{1\leq i\leq s}\ell_i$. For each $i\in\{1,\ldots,s\}$, we find a set of indices $R_i\subseteq\{1,\ldots,m\}$ such that $|R_i|=\ell$ and $\bfA[R_i,S_i]$ contains all nonuniform rows of $\bfA[\{1,\ldots,m\},S_i]$. For $i\in \{1,\ldots,s\}$, we define $\bfB_i=\bfA[R_i,S_i]$.  
For $i\in\{1,\ldots,s\}$, we denote by $\mathbb{0}_i$ and $\mathbb{1}_s$ $\lceil (k+1)/2\rceil\times |S_i|$-matrix with all the elements $0$ and $1$ respectively. We use the matrices $B_i$, $\mathbb{0}_i$ and $\mathbb{1}_i$ as blocks to define 
$$
\bfD=
\left (
\begin{array}{c|c|c|c}
\bfB_1 & \bfB_2 & \cdots & \bfB_s\\
\hline
\mathbb{1} _1& \mathbb{0}_2 & \cdots &\mathbb{0}_s\\
\hline
\mathbb{0} _1& \mathbb{1}_2 & \cdots &\mathbb{0}_s\\
\hline
\vdots & \vdots & \ddots & \vdots\\ 
\hline
\mathbb{0} _1& \mathbb{0}_2 & \cdots &\mathbb{1}_s\\
\end{array}
\right )
$$
We denote the columns of $\bfD$ by $\bfd^1,\ldots,\bfd^n$ following the convention that  the columns of 
$
\left(
\begin{array}{c}
\bfB_i\\ \hline
 \vdots \\ 
\end{array}
\right)
$
are indexed by $j\in S_i$ according to the indexing of the corresponding columns of $\bfA$.

Then our kernelization algorithm returns the instance $(\bfD,r,k)$  of \probAtMostClust.

\medskip 

The correctness of the algorithm is based on the following claim.

\begin{claim}\label{claim:final-kernel}
  $(\bfA,r,k)$ is a \yesinstance of \probAtMostClust if and only if $(\bfD,r,k)$ is a \yesinstance.
\end{claim}

\begin{proof}[Proof of Claim~\ref{claim:final-kernel}]
We show that $\{I_1,\ldots,I_{r'}\}$ is a solution for $(\bfA,r,k)$ if and only if $\{I_1,\ldots,I_{r'}\}$ is a solution for $(\bfD,r,k)$.

Suppose that $\{I_1,\ldots,I_{r'}\}$ is a solution for $(\bfA,r,k)$. Denote by $\bfc^1,\ldots,\bfc^{r'}$ the means of the clusters of the solution obtained by the majority rule. Let $\bfs^1,\ldots,\bfs^{r'}$ be the means of the corresponding clusters for $\bfD$ obtained by the majority rule. By Claim~\ref{claim:sep}, for every $i\in\{1,\ldots,r'\}$, there is $j\in\{1,\ldots,s\}$ such that $I_i\subseteq S_j$. Notice that by the construction of $\bfD$, it holds that for every $p\in S_j$, $\hdist(\bfa^p,\bfc^i)=\hdist(\bfd^p,\bfs^i)$, because all the rows of $\bfA[\{1,\ldots,m\}\setminus R_i,S_i]$ and all the rows of $\bfD[\{\ell+1,\ldots,m\},S_i]$ are uniform. Therefore,
$$ \sum_{i=1}^{r'}\sum_{j\in I_i}\hdist(\bfs^i,\bfd^j)=\sum_{i=1}^{r'}\sum_{j\in I_i}\hdist(\bfc^i,\bfa^j)\leq k$$
and $\{I_1,\ldots,I_{r'}\}$ is a solution for $(\bfD,r,k)$.

Assume that $\{I_1,\ldots,I_{r'}\}$ is a solution for $(\bfD,r,k)$ and denote by $\bfs^1,\ldots,\bfs^{r'}$ the means of the clusters for $\bfD$ obtained by the majority rule.

We observe that for every $i\in\{1,\ldots,r'\}$ there is $j\in\{1,\ldots,s\}$ such that $I_i\subseteq S_j$. To see this, assume that this is not the case and there is a cluster $I_i$ such that there are $p,q\in I_i$ with $p$ and $q$ in distinct sets of the partition $\{S_1,\ldots,S_s\}$. 
Then $\hdist(\bfd^p,\bfd^q)>k$ by the construction of $\bfD$ as $\bfd^p$ and  $\bfd^q$ differ in at least $2\lceil (k+1)/2\rceil$ positions. 
Therefore,
$$ \sum_{j=1}^{r'}\sum_{h\in I_j}\hdist(\bfs^j,\bfd^h)\geq \sum_{h\in I_i}\hdist(\bfs^i,\bfd^h)\geq \hdist(\bfs^i,\bfd^p)+\hdist(\bfs^i,\bfd^q)\geq \hdist(\bfd^p,\bfd^q)>k
$$
contradicting that $\{I_1,\ldots,I_{r'}\}$ is a solution. 

Let $\bfc^1,\ldots,\bfc^{r'}$ be the means of the corresponding clusters for $\bfA$ obtained by the majority rule. Since for every $i\in\{1,\ldots,r'\}$, there is $j\in\{1,\ldots,s\}$ such that $I_i\subseteq S_j$, we again obtain that for every $p\in S_j$, $\hdist(\bfa^p,\bfc^i)=\hdist(\bfd^p,\bfs^i)$. Hence,
$$ \sum_{i=1}^{r'}\sum_{j\in I_i}\hdist(\bfc^i,\bfa^j)=\sum_{i=1}^{r'}\sum_{j\in I_i}\hdist(\bfs^i,\bfd^j)\leq k$$
and $\{I_1,\ldots,I_{r'}\}$ is a solution for $(\bfA,r,k)$.
\end{proof}

Finally, to bound the size $\bfD$, 
recall that $\bfD$ has $n\leq (k+1)(k+r)$ columns and at most $k+r$ of them are pairwise distinct. The matrices $\bfB_1,\ldots,\bfB_s$ have $\ell$ rows where $\ell$ is the maximum number of nonuniform rows in $\bfA[\{1,\ldots,m\},S_i]$. By Claim~\ref{claim:uniform}, 
$$\ell=\max_{1\leq i\leq s}(|S_i|-1)k\leq (n-1)k\leq ((k+1)(k+r)-1)k.$$
Because $s\leq r$, we obtain that $\bfD$ has at most $((k+1)(k+r)-1)k+\lceil(k+1)/2\rceil r$ rows. Therefore, $\bfD$ has $\Oh(k^3(k+r)^2)$ elements. Note also that $\bfD$ has at most 
$((k+1)(k+r)-1)k+r=\Oh(k^2(k+r))$ pairwise distinct rows. This completes the correctness proof.

To evaluate the running time, observe that Reduction Rules~\ref{rule:trivial}--\ref{rule:bigS} demand polynomial time. The greedy algorithm that was used to construct the partition $\mathcal{S}=\{S_1,\ldots,S_s\}$ of $\{1,\ldots,n\}$ is trivially polynomial. The construction of $\bfB_1,\ldots,\bfB_s$ is also polynomial and, therefore, $\bfD$ is constructed in polynomial time.
\end{proof}

\subsection{Ruling out polynomial kernel with parameter $k$.}\label{subsec:nopolyk}
Our next aim is to show that \probAtMostClust parameterized by $k$ does not admit a polynomial kernel unless $\classNP\subseteq\classCoNP/{\rm poly}$. We do this in two steps. First, we use the composition technique introduced by Bodlaender et al.~\cite{BodlaenderDFH09} (see also~\cite{CyganFKLMPPS15} for the introduction to this technique) to show that it is unlikely that the \probCons problem introduced by Boucher, Lo and Lokshtanov in~\cite{BoucherLL11} has a polynomial kernel. Then we use this result to prove the claim for   \probAtMostClust.

\defproblema{\probCons}%
{A (multi) set of $p$ strings $S=\{s_1,\ldots,s_n\}$ of the same length $\ell$ over an alphabet $\Sigma$, a positive integer $r$ and nonengative integer $d$.}%
{Decide whether there is a string $s$ of length $\ell$ over $\Sigma$ and $I\subseteq \{1,\ldots,n\}$ with $|I|=r$ such that $\sum_{i\in I}\hdist(s,s_i)\leq d$.
}

Boucher, Lo and Lokshtanov in~\cite{BoucherLL11} investigated parameterized complexity of \probCons and obtained a number of approximation and inapproximability  results. In particular, they proved that the problem is \classFPT when parameterized by $d$. We show that it is unlikely that this problem has a polynomial kernel for this parameterization.

\begin{theorem}\label{thm:consensus}
\probCons has no polynomial kernel when parameterized by $d$ unless $\classNP\subseteq\classCoNP/{\rm poly}$ even for strings over the binary alphabet. Moreover, the result holds for the instances with $r\leq d$.
\end{theorem}

\begin{proof}
We use the fact that \probCons is \classNP-complete for strings over the binary alphabet $\Sigma=\{0,1\}$~\cite{BoucherLL11} and construct a composition algorithm for the problem parameterized by $d$. 

Let $(S_1,r,d),\ldots,(S_t,r,d)$ be instances of \probCons where $S_1,\ldots,S_t$ are (multi) sets of binary strings of the same length $\ell$. Denote by $\bar{0}$ and $\bar{1}$ the strings of length $d+1$ composed by $0$s and $1$s respectively, that is,
$$\bar{0}=\underbrace{0\ldots0}_{d+1}\text{ and } \bar{1}=\underbrace{1\ldots1}_{d+1}.$$
For $i\in\{1,\ldots,t\}$, we define the set of strings
$$S'_i=\{s\bar{0}^{i-1}\bar{1}\bar{0}^{t-i}\mid s\in S_i\}.$$
Then we put $S^*=\cup_{i=1}^t S_i'$ and consider the instance $(S^*,r,d)$ of \probCons.

We show that $(S^*,r,d)$ is a \yesinstance of \probCons if and only if there is $i\in\{1,\ldots,t\}$ such that $(S_i,r,d)$ is a \yesinstance.

Assume that for $i\in\{1,\ldots,t\}$, the strings of $S_i$ and $S_i'$ are indexed by indices from a set $I_i$, where $I_1,\ldots,I_t$ are disjoint, and denote the strings of $S=\cup_{i=1}^tS_i$ and 
$S^*$ by $s_j$ and $s_j'$ respectively for $j\in\cup_{i=1}^tI_i$.

Suppose that there is $i\in\{1,\ldots,t\}$ such that $(S_i,r,d)$ is a \yesinstance of \probCons. Then there is $I\subseteq I_i$ such that $|I|=r$ and a binary string $s$ of length $\ell$ such that $\sum_{i\in I}\hdist(s,s_i)\leq d$. Let $s'=s\bar{0}^{i-1}\bar{1}\bar{0}^{t-i}$. Since $I\subseteq I_i$, we have that 
\[\sum_{i\in I}\hdist(s',s_i')=\sum_{i\in I}\hdist(s,s_i)\leq d,\] that is, $(S^*,r,d)$ is a \yesinstance.

Assume that $(S^*,r,d)$ is a \yesinstance of \probCons. Then there is  $I\subseteq \cup_{i=}^tI_i$ such that $|I|=r$ and a binary string $s'$ of length $\ell+(d+1)t$ such that $\sum_{i\in I}\hdist(s',s_i')\leq d$. We show that there is $i\in\{1,\ldots,t\}$ such that $I\subseteq I_i$. To obtain a contradiction, assume that there are distinct $i,j\in\{1,\ldots,t\}$ such that $I\cap I_i\neq\emptyset$ and $I\cap I_j\neq\emptyset$. Let $p\in I\cap I_i$ and $q\in I\cap I_j$. We conclude that
\begin{eqnarray*}
\sum_{h\in I}\hdist(s',s_h')&\geq& \hdist(s',s_p')+\hdist(s',s_q')\geq \hdist(s_p',s_q')\\ &\geq& \hdist(\bar{0}^{i-1}\bar{1}\bar{0}^{t-i},\bar{0}^{j-1}\bar{1}\bar{0}^{t-j})\geq 2(d+1)>d,
\end{eqnarray*}
which is a contradiction. Hence, there is $i\in\{1,\ldots,t\}$ such that $I\subseteq I_i$. Let $s$ be a substring of $s'$ containing the first $\ell$ symbols.
We have that 
$$
\sum_{h\in I}\hdist(s,s_h)=\sum_{h\in I}\hdist(s',s_h')\leq d,
$$
that is, $(S_i,r,d)$ is a \yesinstance of \probCons. 

Observe that every string of $S^*$ is  of  length $\ell+(d+1)t$ and that $|S^*|=\sum_{i=1}^t|S_i|$. This means that the size of $(S^*,r,d)$ is polynomial in the sum of the sizes of  $(S_i,r,d)$ and $t$. Note also that the parameter $d$ remains the same. By the results of Bodlaender et al.~\cite{BodlaenderDFH09}, we conclude that \probCons has no polynomial kernel when parameterized by $d$ unless $\classNP\subseteq\classCoNP/{\rm poly}$.

To see that the result holds even if $r\leq d$, we observe that for $r\geq d+1$,  \probCons is  solvable in polynomial time.
Let $(S,r,d)$ be a \yesinstance of \probCons where $S=\{s_1,\ldots,s_n\}$ and $r\geq d+1$. Then there is  a string $s$ and $I\subseteq \{1,\ldots,n\}$ with $|I|=r$ such that $\sum_{i\in I}\hdist(s,s_i)\leq d$. Since $|I|=r\geq d+1$, there is $i\in I$ such that $s=s_i$, that is, the mean string $s$ is one of the input strings. 
This brings us to the following simple algorithm. Let $(S,r,d)$ be an instance of \probCons with  $S=\{s_1,\ldots,s_n\}$ and $r\geq d+1$. For each $i\in \{1,\ldots,n\}$, we check whether the instance has a solution with $s=s_i$. We do it by the greedy selection of $r$ strings closest to $s$ (in  the Hamming distance).  
It is straightforward to see that this is a polynomial time algorithm   solving the problem.
\end{proof}

We use Theorem~\ref{thm:consensus} to obtain our kernelization lower bound for \probAtMostClust.

\begin{theorem}\label{thm:no-kernel}
\probAtMostClust has no polynomial kernel when parameterized by $k$ unless $\classNP\subseteq\classCoNP/{\rm poly}$.
\end{theorem}

\begin{proof}
We reduce \probCons to \probAtMostClust.

Let $(S,r,d)$ be an instance of \probCons,  where $S=\{s_1,\ldots,s_n\}$ is a (multi) set of binary strings of length $\ell$ and $r\leq d$.
Denote by $\bar{0}$ and $\bar{1}$ the strings of length $d+1$ composed by $0$s and $1$s respectively.
 For $i\in\{1,\ldots,n\}$, we set $s_i'=s_i\bar{0}^{i-1}\bar{1}\bar{0}^{n-i}$. We construct the matrix $\bfA$ considering $s_1',\ldots,s_n'$ to be vectors of $\{0,1\}^{\ell+(d+1)n}$ composing the columns of $\bfA$. Slightly abusing notation, we use $s_1',\ldots,s_n'$ to denote the columns of $\bfA$. We set $k=(d+1)r+d$ and set $r'=n-r+1$.

We claim that $(S,r,d)$ is a \yesinstance of \probCons if and only if $(\bfA,r',k)$ is a \yesinstance of \probAtMostClust.

Suppose that $(S,r,d)$ is a \yesinstance of \probCons. Let $I\subseteq\{1,\ldots,n\}$ and a string $s$ of length $\ell$ be a solution, that is, $|I|=r$ and $\sum_{i\in I}\hdist(s,s_i)\leq d$. Denote $s'=s\bar{0}^n$.
By the definition of $s_1',\ldots,s_n'$, we have that
\[\sum_{i\in I}\hdist(s',s_i')\leq d+(d+1)r.\] 
We construct   clusters $I_1,\ldots,I_{r'}$ for $(\bfA,r',k)$ as follows. We put $I_1=I$ and let $I_2,\ldots,I_{r'}$   be one-element disjoint subsets of $\{1,\ldots,n\}\setminus I$. Then we define the means $\bfc^1,\ldots,\bfc^n$ as follows. We set $\bfc^1=s'$ considering $s'$ to be a binary vector. For every single-element cluster $I_i=\{j\}$ for $i\in\{2,\ldots,r'\}$, we set $\bfc^i=s_j'$. Note that 
$\hdist(\bfc^i,s_j)=0$ for $i\in\{2,\ldots,n\}$ and $j\in I_i$.
Then we have 
$$
\sum_{i=1}^{r'}\sum_{j\in I_i}\hdist(\bfc^i,s_j')=\sum_{j\in I}\hdist(\bfc^1,s_j')=\sum_{j\in I}\hdist(s',s_j')\leq d+(d+1)r\leq k,
$$ 
that is, $\{I_1,\ldots,I_{r'}\}$ is a solution for $(\bfA,r',k)$. This means that $(\bfA,r',k)$ is a \yesinstance of \probAtMostClust.

Assume that $(\bfA,r',k)$ is a \yesinstance of \probAtMostClust. Let $\{I_1,\ldots,I_p\}$, $p\leq r'$,  be a solution. Denote by $\bfc^1,\ldots,\bfc^p$ the corresponding means obtained by the majority rule.
If $r=1$, then $(S,r,d)$ is a trivial \yesinstance of \probCons. Let $r\geq 2$.
Then because $r'=n-r+1$, we have that there are clusters $I_j$ with at least two elements. We assume that $|I_j|\geq 2$ for $j\in\{1,\ldots,q\}$ for some $q\leq p$ and $|I_j|=1$ for $j\in \{q+1,\ldots,p\}$. 

We show that $q=1$. Targeting towards  a contradiction, let us assume that $q\geq2$. We have that $\sum_{i=1}^q|I_i|=n-(p-q)$, that is, $n-(p-q)$ columns of $\bfA$ are in clusters with at least two elements. 
By the construction of $\bfA$, we have that the last $n(d+1)$ elements of each mean $\bfc^i$ are $0$s for $i\in\{1,\ldots,q\}$, because the means were constructed by the majority rule. This implies that $\hdist(\bfc^i,s_j')\geq d+1$ for $j\in I_i$ and $i\in \{1,\ldots,q\}$. Clearly, $\hdist(\bfc^i,s_j')=0$ for $j\in I_i$ and $i\in \{q+1,\ldots,p\}$ as these clusters $I_i$ contain one element each. 
Then 
\begin{align*}
\sum_{i=1}^{p}\sum_{j\in I_i}\hdist(\bfc^i,s_j')=&\sum_{i=1}^{q}\sum_{j\in I_i}\hdist(\bfc^i,s_j')\geq \sum_{i=1}^q|I_i|(d+1)=(n-(p-q))(d+1)\\
\geq&(n-r\rq{}+2)(d+1)=(r+1)(d+1)>k\\
\end{align*}
contradicting that $\{I_1,\ldots,I_p\}$ is a solution. Therefore, $q=1$. 

We have that $|I_1|=n-p+1\geq n-r\rq{}+1=r$. Let $I\subseteq I_1$ with $|I|=r$. Recall that we defined  the string $s_j\rq{}=s_j\bar{0}^{j-1}\bar{1}\bar{0}^{n-j}$ for $j\in\{1,\ldots,n\}$ and we consider these strings as the columns of $\bfA$. In particular, the first $\ell$ elements correspond to $s_j$ and the last $n(d+1)$ elements correspond to the string  $\bar{0}^{j-1}\bar{1}\bar{0}^{n-j}$. As above, we have that the last $n(d+1)$ elements of $\bfc^1$ are $0$s. Denote by $s$ the vector composed by the  first $\ell$ elements of $\bfc^1$. We have that 
$$
\sum_{j\in I}\hdist(\bfc^1,s_j')\leq \sum_{j\in I_1}\hdist(\bfc^1,s_j')\leq\sum_{i=1}^{p}\sum_{j\in I_i}\hdist(\bfc^i,s_j')\leq k=(d+1)r+d.
$$
Since 
$$\sum_{j\in I}\hdist(\bfc^1,s_j')=\sum_{j\in I}(\hdist(s,s_j)+(d+1))=\sum_{j\in I}\hdist(s,s_j)+(d+1)r,$$
we conclude that $\sum_{j\in I}(\hdist(s,s_j)\leq d$. This   implies that  $(S,r,d)$ is a \yesinstance of \probCons.

Summarizing, we have 
  constructed a parameterized reduction of \probCons to \probAtMostClust. Notice that for the new parameter $k$ of \probAtMostClust, we have that $k=(d+1)r+d=\Oh(d)$ since $r\leq d$. This observation together with Theorem~\ref{thm:consensus} implies that \probAtMostClust has no polynomial kernel when parameterized by $k$ unless $\classNP\subseteq\classCoNP/{\rm poly}$.
\end{proof}

\section{Subexponential algorithms for \probAtMostClust and  \probFact}\label{sec:subexp}
We have already seen that  \probAtMostClust is solvable in time $2^{\Oh(k\log k)}\cdot (nm)^{\Oh(1)}$ (Theorem~\ref{thm:fpt-clust-k}) and \probFact is solvable in time $2^{\Oh(k\log r)}\cdot(nm)^{\Oh(1)}$ (Proposition~\ref{prop:fact-rk}).
In this section we show that with respect to the combined parameterization by $k$ and $r$, there are algorithms for \probAtMostClust and  \probFact which runs in time subexponential in $k$. For constant rank $r$, the  running times of these algorithms are in 
 $2^{\Oh(\sqrt{k\log k})}\cdot (nm)^{\Oh(1)}$, which  outperforms the algorithms from  Theorem~\ref{thm:fpt-clust-k} and Proposition~\ref{prop:fact-rk}.
 On the other hand,  in both cases we are paying for the improvements on the dependency on $k$ by making the dependency on $r$ worse.  
In Subsection~\ref{sec:fpt-clust} we construct an algorithm for \probAtMostClust  and  in Subsection~\ref{sec:fpt-fact} for \probFact.

\subsection{Subexponential algorithm for \probAtMostClust}\label{sec:fpt-clust}
In this section we design a subexponential  (in $k$) time algorithm for  \probAtMostClust.  

\begin{theorem}\label{thm:clust-subexp}
 \probAtMostClust is solvable in time 
$2^{\Oh(r\sqrt{k\log (k+r)})}\cdot (nm)^{\Oh(1)}$.
\end{theorem}

Towards the proof of Theorem~\ref{thm:clust-subexp}, 
we prove some auxiliary lemmas.

We will be seeking for a special type of solutions. 
\begin{definition}
Let $\bfA$ be an $m\times n$-matrix with rows $\bfa_1,\ldots,\bfa_m$. We say that a vector $\bfc=(c_1,\ldots,c_m)^{\intercal}\in\{0,1\}^m$ \emph{agrees} with $\bfA$ if $c_i=c_j$ whenever $\bfa_i=\bfa_j$ for $i,j\in\{1,\ldots,m\}$. 
\end{definition}
We will be using the following properties of vectors that agree with matrix $\bfA$.

\begin{lemma}\label{lem:agree}
Let $(\bfA,r,k)$ be a \yesinstance of \probAtMostClust. Then $(\bfA,r,k)$ has a solution such that for each cluster of the solution its mean   agrees with $\bfA$.
\end{lemma}

\begin{proof}
Let $\{I_1,\ldots,I_{r'}\}$ be a solution to $(\bfA,r,k)$. For $i\in\{1,\ldots,r'\}$, let $\bfc^i=(c_1^i,\ldots,c_m^i)^{\intercal}\in \{0,1\}^m$ be the mean of the cluster $I_i$ computed by the majority rule. Then if $\bfa_j$ and $\bfa_h$ are rows of $\bfA$ and $\bfa_j=\bfa_h$, then $c_j^i=c_h^i$ because the majority rule computes the same value.  Hence, we have that  $\bfc^i$ agrees with $\bfA$ for each $i\in\{1,\ldots,r'\}$.
\end{proof}

\begin{lemma}\label{lem:agree2} Let 
$\bfA$ be a binary ${m\times n}$  matrix with at most $t$ different rows,  $\bfa$ be a column of $\bfA$ and $h$ be a positive integer. Then there are at most $\sum_{i=1}^h\binom{t}{i}$ binary vectors  
$\bfb\in \{0,1\}^{m}$ agreeing with $\bfA$ and which are within the  Hamming distance at most $h$ from $\bfa$. 
\end{lemma}
\begin{proof}
Let $I_1, \dots, I_t$ be the partition of rows of $\bfA$ into inclusion-maximal sets of equal rows. 
Vector $\bfa$ agrees with $\bfA$. Also for every vector $\bfb$ that agrees with $\bfA$ there is 
$J\subseteq\{1, \ldots, t\}$, such that $\bfb$ is obtained from $\bfa$ by changing for every $i\in J$ the coordinates corresponding to all rows from $I_i$. But since the distance from $\bfa$ and $\bfb$ is at most $h$, the size of $J$ is at most $h$. Hence, the number of such vectors   is at most   $\sum_{i=1}^h\binom{t}{i}$. 
\end{proof}

Now we are ready to prove Theorem~\ref{thm:clust-subexp}. 

\begin{proof}[Proof of Theorem~\ref{thm:clust-subexp}]
Let $(\bfA,r,k)$ be an instance of \probAtMostClust with $A=(\bfa^1,\ldots,\bfa^n)$. 
First, we preprocess the instance using the kernelization algorithm from Theorem~\ref{thm:kernel}. If the algorithm solves the problem, we return the answer and stop.
Assume that this is not the case. Then the algorithm return an instance of  \probAtMostClust where the matrix has at most $k+r$ pairwise distinct columns and $\Oh(k^2(k+r))$ pairwise distinct row. To simplify notations, we use the same notation $(\bfA,r,k)$ for the obtained instance. Denote by $w$ the number of pairwise distinct rows of $\bfA$.

Informally, our algorithm does the following.
For a given  partial clustering of some columns of $\bfA$, budget $d$ and a new  subset $I$ of columns of $\bfA$ which have to be clustered, it tries to 
extend the partial solution by not exceeding the budget $d$.  Some of the columns from $I$ can go to the existing cluster and some can form new clusters. 
Suppose that we know the minimum distance $h$ from vectors in new cluster to their means. Then all vectors which are within the distance less than $h$ to the already existing  
means, can be assigned to the existing clusters. Then we will be basically left with two options. Either the number of columns  to be  assigned to new clusters does not exceed  $\sqrt{d\log w}$; in this case we brute-force in all possible partitions of $I$. Or we can upper bound $h\leq   \sqrt{k/\log w}$ and invoke recursive arguments based on 
Lemma~\ref{lem:agree2}.

Let us give a formal description of the algorithm. 
Towards that we design a recursive   algorithm \textsc{Extend-Means}. 
The input of \textsc{Extend-Means} is 
a set $I\subseteq\{1,\ldots,n\}$, a set of vectors $S\subseteq\{0,1\}^m$  of size at most $r$ 
that agree with $\bfA$, and 
a nonnegative integer $d$,
The ouput of  \textsc{Extend-Means} is
  a set of vectors $C\subseteq\{0,1\}^m$ of size at most $r$  such that each of the vectors   agrees with $\bfA$, $S\subseteq C$, and $\sum_{i\in I}\min\{\hdist(\bfc,\bfa^i)\mid \bfc\in C\}\leq d$ (if such set exists). 
 We say that such a set $C$ is a \emph{solution}. Thus we are looking for a solution extending 
 the partial solution $S$
for the set of column vectors indexed by $I$.

To solve \probAtMostClust, we call $\textsc{Extend-Means}(I=\{1,\ldots,n\}, S=\emptyset,d=k)$. 
The correctness of this step follows from  
Observation~\ref{obs:eq} and Lemma~\ref{lem:agree}.
Algorithm  \textsc{Extend-Means} performs in 4 steps.

\medskip
\noindent
{\bf Step~1.} If $\sum_{i\in I}\min\{\hdist(\bfs,\bfa^i)\mid \bfs\in S\}\leq d$, then $S$ itself satisfies the conditions of the ouput. In this case we return $C=S$ and stop.

\medskip
\noindent
{\bf Step~2.} If $|S|=r$, then we cannot add  vectors to $S$. Return NO and stop.

\medskip
\noindent
{\bf Step~3.} For every $h=0,\ldots,d$, do the following.
\begin{itemize}
\item[(i)] For every $i\in I$, if $\ell=\min\{\hdist(\bfs,\bfa^i)\mid \bfs\in S\}\leq h-1$, then set $I=I\setminus\{i\}$ and $d=d-\ell$.
\item[(ii)] If $|I|\leq \sqrt{d\log w}$, then for each $p\leq \min\{|I|,r-|S|\}$, consider all possible partitions $\{J_0,\ldots,J_p\}$ of $I$, where $J_0$ could be empty, and do the following:
\begin{itemize}
\item for every $j\in \{1,\ldots,p\}$, find the optimal mean $\bfs^j$ for the cluster $J_j$ using the majority rule;
\item set $S=S\cup \{\bfs^1,\ldots,\bfs^p\}$,
\item if $\sum_{i\in I}\min\{\hdist(\bfs,\bfa^i)\mid \bfs\in S\}\leq d$, then return $C=S$ and stop.
\end{itemize}
If we do not return a solution and do not stop for any value of $p$ and choice of $\{J_0,\ldots,J_p\}$, then return NO and stop.
\item[(iii)] If $h\leq d/|I|$, then  for each vector $\bfs\in \{0,1\}^m$   that   agrees with $\bfA$ and  such that $\hdist(\bfs,\bfa^i)=h$ for some $i\in I$   do the following: call  $\textsc{Extend-Means}\{I,S\cup\{\bfs\},d\}$ and if the algorithm returns a solution $C$, then return it and stop.
\end{itemize}

\medskip
\noindent
{\bf Step~4.} Return NO and stop.

\medskip\noindent\textbf{Correctness.} Now argue for the correctness of the algorithm.  First of all, by its  construction, if the algorithm returns a set of vectors $C$, 
 then each of the vectors from  $C$ agrees with $\bfA$, $|C|\leq r$, $S\subseteq C$, and $\sum_{i\in I}\min\{\hdist(\bfc,\bfa^i)\mid \bfc\in C\}\leq d$.

 Now we show that if there is a  solution to $(I,S,d)$
  then the algorithm returns a solution. 
 The proof is by induction on $r-|S|$. We assume that there is a solution to $(I,S,d)$. 
  The base case is when $|S|=r$. Then $C=S$ is a solution and the algorithm returns $C$ in Step~1.

Now we consider the induction step. That is $\vert S\vert <r$. 
By induction hypothesis we have that for any $S'\supset S$ of size at most $r$, $I'\subseteq \{1,\ldots,n\}$ and 
a nonnegative integer $d'$
 such that each vector from $S'$  agrees with $\bfA$,  
  the algorithm  
returns a solution to the input $(S',I',d')$ if such a solution  exists.

If $S$ is a solution, then the algorithm outputs it in Step~2 and we are done.  
Now we assume that $S$ is not a solution. 
Since $S$ is not a solution and there is a solution to $(I,S,d)$ (by assumption), we have that there is a solution $C\supset S$ such that $\sum_{i\in I}\min\{\hdist(\bfc,\bfa^i)\mid \bfc\in C\}$ is minimum and 
for  every  $\bfc\in C\setminus S$, there is $i\in I$ such that $\hdist(\bfs,\bfa^i)\geq \hdist(\bfc,\bfa^i)=h$ for all $\bfs\in S$. 
We choose such a solution $C=C^*$, $\bfc=\bfc^*\in C^*\setminus S$ and $i=i^*\in I$ in such a way that the value of $h$ is the minimized.  Clearly, $h\leq d$. We claim that the algorithm outputs a solution in Step~3 for this value of $h$ unless it already produced a solution for some lesser value of $h$. In the later case we are done.  So now we have that Step 3 is executed for value $h$. 
Let the set $J\subseteq I$ and an integer $d'$ be constructed as follows: 
\begin{itemize}
\item set $J=I$ and $d'=d$,
\item for every $i\in J$, if $\ell=\min\{\hdist(\bfs,\bfa^i)\mid \bfs\in S\}\leq h-1$, then set $J=J\setminus\{i\}$ and $d'=d'-\ell$.
\end{itemize}
Notice that for every $j\in I\setminus J$, $\min\{\hdist(\bfs,\bfa^j)\mid \bfs\in S\}\leq \min\{\hdist(\bfc,\bfa^j)\mid \bfc\in C^*\setminus S\}$. Therefore, $(J,S,d')$ is an equivalent input, that is, we have a solution for this input if and only if there is a solution for  the original input $(I,S,d)$. Moreover, $C^*$ is a solution for $(J,S,d')$.  Observe also that by the choice of $h$, $i^*\in J$. Then for every $\bfc\in C^*$ and $j\in J$, $\hdist(\bfc,\bfa^j)\geq h$.
Note that for the set $I$ and the integer $d$ constructed in Step~3 (i), we have that $I=J$ and $d=d'$.

Suppose that $|J|\leq \sqrt{d'\log w}$. This case is considered in Step~3 (ii). Let $C^*\setminus S=\{\bfc^1,\ldots,\bfc^p\}$. We construct   partition $\{J_0,\ldots,J_p\}$ of $J$ whose sets could be empty as follows. For each $i\in J$, find $t=\min\{\hdist(\bfc,\bfa^i)\mid \bfc\in C^*\}$. If $t=\hdist(\bfc,\bfa^i)$ for $\bfc\in S$, then include $i\in J_0$. Otherwise, find minimum $i\in\{1,\ldots,p\}$ such that $t=\hdist(\bfc^i,\bfa^i)$ and include $i$ in $J_i$. Assume without loss of generality that $J_1,\ldots,J_p$ are nonempty. (Otherwise, we can just ignore empty sets.) For each $i\in\{1,\ldots,p\}$, let $\bfs^i$ be the optimum mean for the cluster $J_i$ constructed by the majority rule. Clearly, $\sum_{j\in J_i}\hdist(\bfs^i,\bfa^j)\leq \sum_{j\in J_i}\hdist(\bfc^i,\bfa^j)$. Recall also that the majority rule constructs vectors that agree with $\bfA$. This implies that $C'=S\cup\{\bfs^1,\ldots,\bfs^p\}$ is a solution. Since in Step~4 we consider all $p\leq \min\{|I|,r-|S|\}$ and all partitions of $I$ into $p+1$ subsets, the algorithm outputs $C'$.
That is the algorithm outputs a solution.

From now we assume that $|J|\geq\sqrt{d'\log w}$. This case is analyzed in Step~3 (iii). 
Note that for every $\bfc\in C^*$ and $j\in J$, $\hdist(\bfc,\bfa^j)\geq h$. Hence, $h\leq d'/|J|$. In Step~3 (iii), for each   vector $\bfs\in \{0,1\}^m$ which agrees with $\bfA$ and which is a the Hamming distance at most $h$ from some  $\bfa^j$, $j\in J$, we call  $\textsc{Extend-Means}(J,S\cup\{\bfs\},d')$. We have that in some branch of the algorithm, we call $\textsc{Extend-Means}(J,S\cup\{\bfc^*\},d')$, because $\hdist(\bfc^*,\bfa^{i^*})=h$ and $i^*\in J$. By the inductive assumption, the algorithm outputs a solution $C'$ for the input $(J,S\cup\{\bfc^*\},d')$. Then $C'$ is a solution to $(I,S,d)$. 

To complete the correctness proof, note that the depth of the recursion is upper bounded by $r$. It follows that the algorithm perform finite number of steps and returns either a solution or the answer NO.

\medskip\noindent\textbf{Running time.}
To evaluate the running time, note that Steps~1 and 2 can be done in polynomial time. In Step~3, we consider $d+1\leq k+1$ values of $h$ and for each $h$, we  perform Steps~3 (i)--(iii). Step~3 (i) is done in polynomial time. Since $|I|\leq \sqrt{d\log w}$ and $p\leq r$ in Step~3 (ii), we consider  $2^{\Oh(\log r\sqrt{k\log w})}$ partitions of $I$ in this step. 
Because $w=\Oh(k^2(k+r))$, Step~3 (ii) can be done in time $2^{\Oh(\log r\sqrt{k\log (k^2(k+r))})}\cdot (nm)^{\Oh(1)}$.
 In Step~3 (iii), we have that $h\leq d/|I|\leq \sqrt{d/\log w}\leq \sqrt{k/\log w}$. Recall that $\bfA$ has at most $r+k$ pairwise-distinct columns.
Hence, to construct $\bfs$ in Step~3 (iii), we consider at most $r+k$ columns $\bfa^i$ for $i\in I$. Recall also that $\bfA$ has $w$ distinct rows. Since $\bfs$ agrees with $\bfA$, by Lemma~\ref{lem:agree2},  we have that there are at most $2^{\Oh(\log w\sqrt{k/\log w}}$ vectors $\bfs$ at the Hamming distance at most $h$ from $\bfa^i$. It follows that Step~3 (iii) without recursive calls of \textsc{Extend-Means} can be performed in time $2^{\Oh(\log w\sqrt{k/\log w})}\cdot (nm)^{\Oh(1)}$ and we have  $2^{\Oh(\sqrt{k\log w})}$ recursive calls of the algorithm. The depth of the recursion is upper bounded by $r$. Using the property that $w=\Oh(k^2(k+r))$, we have that  the total running time of our algorithm is $2^{\Oh(r\sqrt{k\log k^2(k+r)})}\cdot (nm)^{\Oh(1)}$ or $2^{\Oh(r\sqrt{k\log (k+r)})}\cdot (nm)^{\Oh(1)}$. 
\end{proof}

\subsection{Subexponential algorithm for \probFact}\label{sec:fpt-fact}
In this subsection we prove the following theorem. 

 \begin{theorem}\label{thm:fact-subexp}
 \probFact is solvable in time $2^{\Oh(r^{3/2}\sqrt{k\log k})}\cdot (nm)^{\Oh(1)}$.
\end{theorem}

The general idea of parameterized subexponential  time algorithm for  \probFact\ (Theorem~\ref{thm:fact-subexp})is  similar to the parameterized subexponential  time algorithm for \probAtMostClust. However, 
because we cannot use the majority rule like the case of \probAtMostClust, there are some complications.
We start with some auxiliary lemmas and then prove Theorem~\ref{thm:fact-subexp}. 

\begin{lemma}\label{lem:diff-fact}
Let $(\bfA,r,k)$ be a \yesinstance of \probFact. Then $\bfA$ has at most $2^r+k$ pairwise-distinct columns and at most $2^r+k$ pairwise-distinct rows.
\end{lemma}

\begin{proof}
Since $(\bfA,r,k)$ is a \yesinstance of \probFact, there is an  $m\times n$-matrix $\bfB$ over \GF{} with $\rank(\bfB)\leq r$ such that $\|\bfA-\bfB\|_F^2\leq k$. By Observation~\ref{obs:rank}, $\bfB$ has at most $2^r$ pairwise-distinct columns and at most $2^r$ pairwise-distinct rows. Since $\|\bfA-\bfB\|_F^2\leq k$, we have that  matrices $\bfA$ and $\bfB$ differ in at most $k$ columns and in at most $k$ rows. This immediately implies the claim.
\end{proof}

In this section it is more convenient   to use the alternative formulation  of \probFact from Observation~\ref{obs:eq-fact}: given 
an $m\times n$-matrix $\bfA$ with columns $\bfa^1, \ldots, \bfa^n$ over \GF, a positive integer $r$ and a nonnegative integer $k$,
decide whether there is a positive integer $r'\leq r$  and linearly independent vectors $\bfc^1,\ldots,\bfc^{r'}\in\{0,1\}^m$ over \GF{} such that 
$\sum_{i\in 1}^n\min\{\hdist(\bfs,\bfa^i)\mid \bfs=\bigoplus_{j\in I}\bfc^j,~I\subseteq \{1,\ldots,r'\}\}\leq k$.
Respectively, throughout this section we say that a set of vectors  $C=\{\bfc^1,\ldots,\bfc^{r'}\}$ satisfying the above condition is a \emph{solution} for $(\bfA,r,k)$. 
Recall that a vector $\bfc=(c_1,\ldots,c_m)^{\intercal}\in\{0,1\}^m$ 
agrees with an $m\times n$-matrix $\bfA$ 
 with rows $\bfa_1,\ldots,\bfa_m$ if  $c_i=c_j$ whenever $\bfa_i=\bfa_j$ for $i,j\in\{1,\ldots,m\}$.

\begin{lemma}\label{lem:agree-fact}
Let $(\bfA,r,k)$ be a \yesinstance of \probFact. Then $(\bfA,r,k)$ has a solution $C=\{\bfc^1,\ldots,\bfc^{r'}\}$ such that $\bfc^i$ agrees with $\bfA$ for all $i\in\{1,\ldots,r'\}$.
\end{lemma}

\begin{proof}
Denote by $\bfa^1,\ldots,\bfa^n$ and $\bfa_1,\ldots,\bfa_m$ the columns and rows of $\bfA$, respectively. Let $C=\{\bfc^1,\ldots,\bfc^{r'}\}$, where $\bfc^i=(c_1^i,\ldots,c_m^i)^\intercal$ for $i\in\{1,\ldots,r'\}$, be a solution such that the total number of pairs of integers $p,q\in \{1,\ldots,m\}$ such that $\bfa_p=\bfa_q$ and there is $i\in\{1,\ldots,n\}$ such that $c_p^i\neq c_q^i$ is minimum.
We claim that each $\bfc^i$ agrees with $\bfA$. To obtain a contradiction, assume that there are $p,q\in \{1,\ldots,m\}$ such that $\bfa_p=\bfa_q$ and there is $i\in\{1,\ldots,n\}$ such that $c_p^i\neq c_q^i$. 
For $i\in\{1,\ldots,n\}$, denote by $J(i)\subseteq \{1,\ldots,n\}$ the subset of indices such that 
$$\hdist(\bigoplus_{j\in J(i)}\bfc^j,\bfa^i)=\min\{\hdist(\bfs,\bfa^i)\mid \bfs=\bigoplus_{j\in I}\bfc^j,~I\subseteq \{1,\ldots,r'\}\}.$$
For each $j\in \{1,\ldots,m\}$, let $\bfs^j=(s_1^j,\ldots,s_n^j)$ where $s_i^j=\bigoplus_{h\in J(i)}c_j^h$ for $i\in\{1,\ldots,n\}$. Assume without loss of generality that $\hdist(\bfs^p,\bfa_p)\leq \hdist(\bfs^q,\bfa_q)$. For $i\in \{1,\ldots,r'\}$, denote by $\hat{\bfc}^i$ the vector obtained from $\bfc^i$ by  replacing  $c_q^i$ with $c_p^i$. We have that
$$\sum_{i=1}^n\hdist(\bigoplus_{h\in I(i)}\bfc^h,\bfa^i)\leq \sum_{i=1}^n\hdist(\bigoplus_{h\in I(i)}\hat{\bfc}^h,\bfa^i).$$
This implies that the set of vectors $\hat{C}$ obtained from $\{\hat{\bfc}^1,\ldots,\hat{\bfc}^{r'}\}$ by taking a maximum set of linearly independent vectors,  is a solution. But this contradicts the choice of $C$, because $\hat{c}_p^i=\hat{c}_q^i=c_p^i$ for $i\in\{1,\ldots,r'\}$.
We conclude that the vectors of $C$ agree with $\bfA$.
\end{proof}

Now we are ready to prove Theorem~\ref{thm:fact-subexp}. 

\begin{proof}[Proof of Theorem~\ref{thm:fact-subexp}]
Let $(\bfA,r,k)$ be an instance of \probFact with $A=(\bfa^1,\ldots,\bfa^n)$. 
First, we preprocess the instance using Lemma~\ref{lem:diff-fact}. That is,  if $\bfA$ has at least $2^r+k+1$ pairwise-distinct columns or at least $2^r+k+1$ distinct rows, we return the answer NO and stop.  Now on we assume that the number of pairwise-distinct columns as well as rows is at most $2^r+k$.

As in the proof of Theorem~\ref{thm:clust-subexp}, we construct an algorithm extending a partial solution. 
Towards that we design a recursive  algorithm \textsc{Extend-Solution}. 
 Input of \textsc{Extend-Solution} is 
 a set $I\subseteq\{1,\ldots,n\}$,  a $p$-sized set of linearly independent vectors 
$S=\{\bfs^1,\ldots,\bfs^p\}\subseteq\{0,1\}^m$ over \GF{} 
that agree with $\bfA$,  and  a nonnegative integer $d$. 
 The output of \textsc{Extend-Solution} is 
 a set of linearly independent vectors $C=\{\bfc^1,\ldots,\bfc^{r'}\}\subseteq\{0,1\}^m$ over \GF, such that each of them   agrees with $\bfA$,   $|C|\leq r$, $S\subseteq C$ and $\sum_{i\in I}\min\{\hdist(\bfs,\bfa^i)\mid \bfs=\bigoplus_{j\in J}\bfc^j,~J\subseteq\{1,\ldots,r'\}\}\leq d$ if it exists or 
 concludes that no such set exists. 
 We say that such a set $C$ is a \emph{solution} and call $(I,S,d)$ an \emph{instance} of \textsc{Extend-Solution}. Then to solve  \probFact, we call $\textsc{Extend-Solution}(\{1,\ldots,n\},\emptyset,k)$. For the  simplicity of explanation, we solve the decision version of the problem. Our algorithm could be easily modified to produce a solution if it exists.

We denote by $\bfS=(\bfs^1,\ldots,\bfs^p)$ the $m\times p$-matrix whose columns are the vectors of $S$ and denote by $\bfs_1,\ldots,\bfs_m$ its rows. For $I\subseteq \{1,\ldots,n\}$, we define
$\bfA^I=\bfA[\{1,\ldots,m\},I]$. We denote by $\bfa_1^I,\ldots,\bfa_m^{I}$ the rows of $\bfA^I$. Our algorithm uses the following properties of $\bfS$ and $\bfA^I$.

\begin{claim}
\label{claim:thm3}
A solution $C$ for an instance $(I,S,d)$ exists if and only if there is  $r'\leq r$ and two $r'$-tuples of vectors $(\bfx^1,\ldots,\bfx^{r'})$ and $(\bfy^1,\ldots,\bfy^{r'})$ with $\bfx^i\in\{0,1\}^p$ and $\bfy^i\in\{0,1\}^{|I|}$ for $i\in\{1,\ldots,r'\}$ such that 
$$\sum_{i=1}^m\min\{\hdist(\bfy,(\bfa^I_i)^\intercal)\mid \bf\bfy=\bigoplus_{j\in J}\bfy^j,~J\subseteq\{1,\ldots,r'\}\text{ and }\bfs_i^\intercal=\bigoplus_{j\in J}\bfx^j\}\leq d.$$
\end{claim}

\begin{proof}[Proof of Claim~\ref{claim:thm3}]
For an $m\times p$-matrix $\bfX$ and an $m\times q$-matrix $\bfY$, we denote by $(\bfX|\bfY)$ the \emph{augmentation} of $\bfX$ by $\bfY$, that is, the $m\times (p+q)$-matrix whose first $p$ columns are the columns of $\bfX$ and the last $q$ columns are the columns of $\bfY$.

Note that we have a solution $C$ for an instance $(I,S,d)$ if and only if there is an $m\times |I|$-matrix $\hat{\bfA}^I$ such that $\rank(\bfS|\hat{\bfA}^I)\leq r$ and 
$\|\bfA^I-\hat{\bfA}^I\|^2_F \leq d$. This observation immediately implies the claim.
\end{proof}

Now we are ready to describe algorithm \textsc{Extend-Solution}. 
Let $(I,S,d)$ be the input of \textsc{Extend-Solution}.
It performs in 4 steps.

\medskip
\noindent
{\bf Step~1.} If $\sum_{i\in I}\min\{\hdist(\bfs,\bfa^i)\mid \bfs=\bigoplus_{j\in J}\bfs^j,~J\subseteq\{1,\ldots,p\}\}\leq d$, then return YES  and stop.

\medskip
\noindent
{\bf Step~2.} If $p=r$, then return NO and stop.

\medskip
\noindent
{\bf Step~3.} For every $h=0,\ldots,d$, do the following.
\begin{itemize}
\item[(i)] For every $i\in I$, if $\ell=\min\{\hdist(\bfs,\bfa^i)\mid  \bfs=\bigoplus_{j\in J}\bfs^j,~J\subseteq\{1,\ldots,p\}\}\leq h-1$, then set $I=I\setminus\{i\}$ and $d=d-\ell$.
\item[(ii)] If $|I|\leq \sqrt{d\log(2^r+d)}$, then do the following:
\begin{itemize}
\item for each $r'\leq r$, consider all possible $r'$-tuples of vectors $(\bfx^1,\ldots,\bfx^{r'})$ and $(\bfy^1,\ldots,\bfy^{r'})$ with $\bfx^i\in\{0,1\}^p$ and $\bfy^i\in\{0,1\}^{|I|}$ for $i\in\{1,\ldots,r'\}$, and do the following:
if $\sum_{i=1}^m\min\{\hdist(\bfy,(\bfa_i^I)^\intercal)\mid \bfy=\bigoplus_{j\in J}\bfy^j,~J\subseteq\{1,\ldots,r'\}\text{ and }\bfs_i^\intercal=\bigoplus_{j\in J}\bfx^j\}\leq d$, 
then return YES and stop.
\item Return NO and stop.
\end{itemize}
\item[(iii)] If $h\leq d/|I|$, then  consider each vector $\bfs\in \{0,1\}^m$ such that $\bfs$ agrees with $\bfA$, $\bfs$ is linearly independent with the vectors of $S$ and $\hdist(\bfs,\bfa^i)=h$ for some $i\in I$ and do the following: call  $\textsc{Extend-Solution}(I,S\cup\{s\},d)$ and if the algorithm returns YES, then return YES and stop.
\end{itemize}
\medskip
\noindent
{\bf Step~4.} Return NO and stop.

\medskip\noindent\textbf{Correctness.} 
The construction of the algorithm and Claim~\ref{claim:thm3} imply that if the algorithm returns YES, then there is a solution $C$ for the instance $(I,S,d)$.

Now we show the reverse direction. That is, we show that if there is a solution $C$ for $(I,S,d)$, then the algorithm returns YES. We assume that there is a solution $C$ for $(I,S,d)$. The proof is by induction on $r-|S|$. The base case is when $\vert S\vert=r$. In this case $C=S$ and algorithm returns YES in Step~1. 

Now consider the induction step. That is, $\vert S\vert < r$. By the induction hypothesis we have 
that for every $I'\subseteq\{1,\ldots,n\}$, a set of size at most $r$ linearly independent vectors $S'\supset S$ that agree with $\bfA$, and  a nonnegative integer $d'$, the algorithm 
returns YES if a solution for the input $(S',I',d')$ exists. If the algorithm output YES in Step~1, then we are done. 
Otherwise we have that $S$ is not a solution. Since $\vert S\vert<r$, Step~2 is not executed.  
Since $S$ is not a solution, we have that  for every solution $C=\{\bfc^1,\ldots,\bfc^{r'}\}\supset S$ minimizing the sum 
$\sum_{i\in I}\min\{\hdist(\bfc,\bfa^i)\mid \bfc=\bigoplus_{j\in J}\bfc^j,~J\subseteq \{1,\ldots,r'\}\}$  satisfies the following property:  
There exists $i\in I$ such that for every linear combination $\bfs\in\{0,1\}^m$ of vectors from $S$ we have
$\hdist(\bfs,\bfa^i)> h$, where $h=\min\{\hdist(\bfc,\bfa^i)\mid \bfc=\bigoplus_{j\in J}\bfc^j,~J\subseteq \{1,\ldots,r'\}\}$.
 We choose a solution $C=C^*$ and $i=i^*\in I$ in such a way that the value of $h$ is minimum. 
Notice that there is a subset $J^*\subseteq \{1,\ldots,r'\}$ such that 
$\bfc^*=\bigoplus_{j\in J^*}\bfc^j$,  $h=\hdist(\bf\bfc^*,\bfa^{i^*})$ and 
$S\cap \{\bfc^j\colon j\in J^*\}\neq \emptyset$. 
We assume that $\bfc^*\in C^*$.  
Otherwise, we just add $\bfc^*$ to the solution and exclude arbitrary $\bfc^j\notin S$ with $j\in J^*$ from $C^*$ (because $\bfc^*=\bigoplus_{j\in J^*}\bfc^j$). Let $C^*=\{\bfc^1,\ldots,\bfc^{r'}\}$.  
Clearly, $h\leq d$. We claim that the algorithm outputs YES in Step~3 for this value of $h$ unless it already produced the same answer for some smaller value of $h$. 
Let the set $I'\subseteq I$ and the integer $d'$ be constructed as follows: 
\begin{itemize}
\item set $I'=I$ and $d'=d$,
\item for every $i\in I'$, if $\ell=\min\{\hdist(\bfs,\bfa^i)\mid  \bfs=\bigoplus_{j\in J}\bfs^j,~J\subseteq \{1,\ldots,p\}\}\leq h-1$, then set $I'=I'\setminus\{i\}$ and $d'=d'-\ell$.
\end{itemize}
Notice that for every $j\in I\setminus I'$, \[\min\{\hdist(\bfs,\bfa^j)\mid \bfs=\bigoplus_{j\in J}\bfs^j,~J\subseteq \{1,\ldots,p\}\}
\leq \min\{\hdist(\bfc,\bfa^j)\mid \bfc=\bigoplus_{j\in J}\bfc^j,~J\subseteq \{1,\ldots,r'\}\}.\]

Because of the choice of $h$, 
we have a solution for $(I',S,d')$ if and only if there is a solution for  the original input $(I,S,d)$. Moreover, $C^*$ is a solution for $(I',S,d')$.  Observe also that by the choice of $h$ we have that $i^*\in I'$.

Note that for the set $I$ and an integer $d$ constructed in Step~3 (i), we have that $I=I'$ and $d=d'$.
Suppose that $|I'|\leq \sqrt{d'\log (2^r+d')}$. This case is considered in Step~3 (ii). By Claim~\ref{claim:thm3}, there is  $r''\leq r$ and two $r''$-tuples of vectors $(\bfx^1,\ldots,\bfx^{r''})$ and $(\bfy^1,\ldots,\bfy^{r''})$ with $\bfx^i\in\{0,1\}^p$ and $\bfy^i\in\{0,1\}^{|I'|}$ for $i\in\{1,\ldots,r''\}$ such that 
\begin{equation}\label{eq:check}
\sum_{i=1}^m\min\{\hdist(\bfy,\bfa_i^\intercal)\mid \bfy=\bigoplus_{j\in J}\bfy^j,~J\subseteq\{1,\ldots,r''\}\text{ and }\bfs_i^\intercal=\bigoplus_{j\in J}\bfx^j\}\leq d'.
\end{equation}
Since our algorithm considers all $r''\leq r$ and all possible $r''$-tuples of vectors $(\bfx^1,\ldots,\bfx^{r''})$ and $(\bfy^1,\ldots,\bfy^{r''})$ with $\bfx^i\in\{0,1\}^p$ and $\bfy^i\in\{0,1\}^{|I'|}$ for $i\in\{1,\ldots,r''\}$ and verifies (\ref{eq:check}), we obtain that the algorithm outputs YES in Step~3 (ii).

From now we assume that $|I'|>\sqrt{d'\log (2^r+d')}$. This case is analyzed in Step~3 (iii).
Note that for every $j\in I'$ and $J\subseteq\{1,\ldots,r'\}$, $\hdist(\bfc,\bfa^j)\geq h$ for $\bfc=\bigoplus_{i\in J}\bfc^i$. Hence, $h\leq d'/|I'|$. 
In Step~3 (iii), we consider every vector $\bfs\in \{0,1\}^m$ such that $\bfs$ agrees with $\bfA$, $\bfs$ is linearly independent with the vectors of $S$ and $\hdist(\bfs,\bfa^j)=h$ for some $j\in I'$, and call  $\textsc{Extend-Solution}(I',S\cup\{s\},d')$. We have that in some recursive call of the algorithm, we call $\textsc{Extend-Solution}(J,S\cup\{\bfc^*\},d')$, because $\hdist(\bfc^*,\bfa^{i^*})=h$ and $i^*\in I'$. By the induction hypothesis, the algorithm outputs YES for the input $(I',S\cup\{\bfc^*\},d')$. Then we output YES for  $(I,S,d)$. 

To complete the correctness proof, note that the depth of the recursion is at most $r$. It follows that the algorithm performs a finite number of steps and correctly reports that $(I,S,d)$ is a \yesinstance or a no-instance.

\medskip\noindent\textbf{Running time.} 
To evaluate the running time, note that Steps~1 and 2 can be done in polynomial time. In Step~3, we consider $d+1\leq k+1$ values of $h$ and for each $h$, perform Steps~3 (i)--(iii). 
Because $p\leq r$, Step~3 (i) can be done in time $2^{\Oh(r)}\cdot(nm)^{\Oh(1)}$. 
Since $r'\leq r$, $|I|\leq \sqrt{d\log (2^r+d)}\leq \sqrt{k\log(2^r+k)}$ and $p\leq r$ in Step~3 (ii), we consider  $2^{\Oh(r(r+\sqrt{k\log (2^r+k)}))}$ $r'$-tuples of vectors $(\bfx^1,\ldots,\bfx^{r'})$ and $(\bfy^1,\ldots,\bfy^{r'})$. It implies that Step~3 (ii)  takes time $2^{\Oh(r(r+\sqrt{k\log (2^r+k)}))}\cdot (nm)^{\Oh(1)}$.
In Step~3 (iii), we have that $h\leq d/|I|\leq \sqrt{d/\log (2^r+d)}\leq \sqrt{k/\log (2^r+k)}$. Recall that $\bfA$ has at most $2^r+k$ pairwise-distinct columns.
Hence, to construct $\bfs$ is Step~3 (iii), we consider at most $2^r+k$ columns $\bfa^i$ for $i\in I$. 
Recall also that $\bfA$ has at most $2^r+k$ distinct rows. Since $\bfs$ agrees with $\bfA$, by Lemma~\ref{lem:agree2}, there are at most $2^{\Oh(\log (2^r+k)\sqrt{k/\log(2^r+k)}}$ vectors $\bfs$ at distance $h$ from $\bfa^i$. It follows that Step~3 (iii) without recursive calls of \textsc{Extend-Solution} can be done in time $2^{\Oh(\sqrt{k\log (2^r+k)})}\cdot (nm)^{\Oh(1)}$ and we have  $2^{\Oh(\sqrt{k\log (2^r+k)})}$ recursive calls of the algorithm. The depth of the recursion is at most $r$. Therefore, the total running time of our algorithm is $2^{\Oh(r\sqrt{k\log (2^r+k)})}\cdot (nm)^{\Oh(1)}$ which is upper bounded by   $2^{\Oh(r^{3/2}\sqrt{k\log k})}\cdot (nm)^{\Oh(1)}$
\end{proof}

\section{Subexponential algorithms for \probPClust and \probBFact}\label{sec:p-clust}
In this section we give  a parameterized subexponential algorithm for \probBFact. This algorithm is more complicated than the one for \probFact. The main reason to that is that now the  elements of the matrices along with the boolean operations do not form a field and thus many nice properties of matrix-rank cannot be used here. The way we handle this issue is to solve the \probPClust problem. As far as we obtain  a subexponential algorithm for  
\probPClust, a simple reduction will provide an algorithm for \probBFact.

\medskip

Let $(\bfA,\bfP,k)$ be an instance of \probPClust. We say that a matrix $\bfB$ is a \emph{solution} for the instance if $\|\bfA-\bfB\|^2_F\leq k$ and $\bfB$ is a $\bfP$-matrix.%
We need the following lemma.

\begin{lemma}\label{lem:diff-pclust}
Let $(\bfA,\bfP,k)$ be a \yesinstance of \probPClust where $\bfP$ is a $p\times q$-matrix. Then $\bfA$ has at most $p+k$ pairwise-distinct rows and at most $q+k$ pairwise-distinct columns.
\end{lemma}

\begin{proof}
Let $\bfB$ be a solution to $(\bfA,\bfP,k)$. Then by Observation~\ref{obs:p-matr-numberofrows}, $\bfB$ has at most $p$ pairwise-distinct rows and at most $q$ pairwise-distinct columns. Since $\bfA$ and $\bfB$ differ in at most $k$ elements, the claim follows. 
\end{proof}

Similarly to the algorithms for \probClust and \probFact in Theorems~\ref{thm:clust-subexp} and \ref{thm:fact-subexp} respectively, we construct a  recursive branching  algorithm for \probPClust. 
Recall that in the algorithms for  \probClust and \probFact,  we   solved  auxiliary problems. In these auxiliary problems one has to   extend a partial solution while reducing the set of ``undecided'' columns. In particular, if the set of these undecided columns is sufficiently  small, we use brute force algorithms to solve the problems. Here we use a similar strategy, but the auxiliary extension problem is slightly more complicated. 
Following is the auxiliary  extension problem.

 \defproblema{\textsc{\probPClustExt}}%
{An $m\times n$ binary matrix $\bfA$,  
a pattern $p\times q$-matrix $\bfP$, 
a partition $\{X,Y,Z\}$ of $\{1,\ldots,n\}$, where some sets could be empty, such that $|X|+|Y|=q$,
 and a nonnegative integer $k$.}%
{Decide whether there is an $m\times n$-matrix $\bfB$ such that i)~$\bfB[\{1,\ldots,m\},X]=\bfA[\{1,\ldots,m\},X]$, ii) $\bfB$ and $\bfB[\{1,\ldots,m\},X\cup Y]$ are $\bfP$-matrices, and iii) $\|\bfA-\bfB\|_F^2\leq k$.
}
We call a matrix $\bfB$ satisfying i)--iii) \emph{solution} for \probPClustExt.
The following lemma will be used in the algorithm when the sum  $|Y|+|Z|$ is sufficiently small.

\begin{lemma}\label{lem:ext-Pclust}
\probPClustExt is solvable in time $2^{\Oh(p(r+\log k)+p\log p+q\log q)}\cdot (nm)^{\Oh(1)}$, where $r=|Y|+|Z|$.
\end{lemma}

\begin{proof}
Let $(\bfA,\bfP,\{X,Y,Z\},k)$ be an instance of \probPClustExt. Let $ \bfa^1,\ldots,\bfa^n$ be the columns   and  $\bfa_1,\ldots,\bfa_m$ be the rows of matrix $\bfA=(a_{ij})\in \{0,1\}^{m\times n}$.  

If $\bfA$ has at least $p+k+1$ pairwise-distinct rows or $q+k+1$ pairwise-distinct columns we return NO,  and stop.  This is correct by Lemma~\ref{lem:diff-pclust}.  Also if $m<p$ or $n<q$, we have a trivial no-instance and we again can safely return NO and stop. From now we assume that $m\geq p$, $n\geq q$, and that $\bfA$ has at most $p+k$ pairwise-distinct rows and at most $q+k$ pairwise-distinct columns.

Suppose that matrix $\bfB$ with the rows $\bfb_1,\ldots,\bfb_m$ is a solution to $(\bfA,\bfP,\{X,Y,Z\},k)$. We say that a set $I\subseteq\{1,\ldots,m\}$ \emph{represents} $\bfP$ with respect to $\bfB$ if 
a) $|I|=p$, b) $\bfB[I,\{1,\ldots,n\}]$ and $\bfB[I,X\cup Y]$ are $\bfP$-matrices, and c) for every $i\in\{1,\ldots,m\}\setminus I$, there is $j\in I$ such that $\bfb_i=\bfb_j$.
Clearly, for every solution $\bfB$, there is $I\subseteq\{1,\ldots,m\}$ that represents $\bfP$ with respect to $\bfB$.

We say that two sets of indices $I,I'\subseteq \{1,\ldots,m\}$ are \emph{equivalent} with respect to $\bfA$ if the matrices $\bfA[I,\{1,\ldots,n\}]$ and $\bfA[I',\{1,\ldots,n\}]$ are isomorphic. 
Observe that if $I,I'$ are equivalent with respect to $\bfA$, then $(\bfA,\bfP,\{X,Y,Z\},k)$ has a solution $\bfB$ such that $I$ represents $\bfP$ with respect to $\bfB$ if and only if the same instance has a solution 
$\bfB'$ such that $I'$ represents $\bfP$ with respect to $\bfB'$. 

Since $\bfA$ has at most $p+k$ pairwise-distinct rows, there are at most $(p+k)^p$ pairwise nonequivalent with respect to $\bfA$ sets of indices $I\subseteq\{1,\ldots,m\}$ of size $p$. We consider such sets and for each $I$, we check whether there is a solution $\bfB$ for  $(\bfA,\bfP,\{X,Y,Z\},k)$ with $I$ representing $\bfP$ with respect to $\bfB$. If we find that there is a solution, we return YES and stop, and we return NO if there is no solution for every choice of $I$. From now we assume that $I$ is given.

Our aim now is to check the existence of a solution $\bfB =(b_{ij})\in \{0,1\}^{m\times n}$ with $I$ representing $\bfP$ with respect to $\bfB$. 
 We denote by $\bfb_1,\ldots,\bfb_m$ the rows of $\bfB$.

We consider all possible matrices $\bfB[I,\{1,\ldots,n\}]$. Recall that for each $i\in \{1,\ldots,m\}$, we should have that $b_{ij}=a_{ij}$ for $j\in X$. It implies that there are at most $2^{|Y|+|Z|}$ possibilities to restrict $\bfb_i$ for $i\in I$, and there are at most $2^{(|Y|+|Z|)p}$ possible matrices $\bfB[I,\{1,\ldots,n\}]$. Since the  matrix should satisfy the condition b),  we use Observation~\ref{obs:p-matr-brute} to check this condition in time $2^{\Oh(p\log p+q\log q)}\cdot(nm)^{\Oh(1)}$. If it is violated, we discard the current choice of  $\bfB[I,\{1,\ldots,n\}]$. Otherwise, we check whether $\bfB[I,\{1,\ldots,n\}]$ can be extended to a solution $\bfB$ satisfying c). For $i,j\in \{1,\ldots,m\}$, we define
$$\hdist^*(\bfb_i,\bfa_j)=
\begin{cases}
\hdist(\bfb_i,\bfa_j)&\mbox{ if } b_{ih}=a_{jh}\text{ for }h\in X,\\
+\infty&\mbox{ otherwise.} 
\end{cases}
$$
Then we observe that $\bfB[I,\{1,\ldots,n\}]$ can be extended to a solution $\bfB$ satisfying c) if and only if
$$
\sum_{i\in I}\hdist(\bfb_i,\bfa_i)+\sum_{j\in\{1,\ldots,m\}\setminus I}\min\{\hdist^*(\bfb_i,\bfa_j)\mid i\in I\}\leq k.
$$
We verify the condition and return YES and stop if it holds. Otherwise we discard the current choice of $\bfB[I,\{1,\ldots,n\}]$. If we fail for all choices of $\bfB[I,\{1,\ldots,n\}]$, we return NO  and stop.

It remains to evaluate the running time. We can check in polynomial time whether $\bfA$ has at most $p+k$ rows and at most $q+k$ columns. Then we construct at most $(p+k)^p$ pairwise nonequivalent (with respect to $\bfA$) sets of indices $I\in\{1,\ldots,m\}$ of size $p$ and for each $I$, we consider at most  $2^{(|Y|+|Z|)p}$ possible matrices $\bfB[I,\{1,\ldots,n\}]$. Then for each choice of matrix $\bfB[I,\{1,\ldots,n\}]$, we first check in time $2^{\Oh(p\log p+q\log q)}\cdot(nm)^{\Oh(1)}$ whether this matrix satisfies b) and then 
check in polynomial time whether $\bfB[I,\{1,\ldots,n\}]$ can be extended to a solution. We obtain that the total running time is $2^{\Oh(p(\log k+|Y|+|Z|)+p\log p+q\log q)}\cdot(nm)^{\Oh(1)}$.
\end{proof}

Now we are ready to prove the main technical result of this section.
 \begin{theorem}\label{thm:pclust-subexp}
 \probPClust is  solvable in
 $2^{\Oh((p+q)\sqrt{k\log (p+k)}+p\log p+q\log q+q\log(q+k))}\cdot (nm)^{\Oh(1)}$  time.
\end{theorem}

\begin{proof}
Let $(\bfA,\bfP,k)$ be an instance of \probPClust. Denote by $\bfa_1,\ldots,\bfa_m$ and $\bfa^1,\ldots,\bfa^n$ the rows and columns of $\bfA$ respectively, and let $\bfP$ be a $p\times q$-matrix. 

First, we preprocess the instance using Lemma~\ref{lem:diff-pclust}. If $\bfA$ has at least $p+k+1$ pairwise-distinct rows or at least $q+k+1$ pairwise-distinct columns, we return NO and stop. 
We do the same if $m<p$ or $n<q$.
Assume from now that this is not the case, that is, $p\leq m$, $q\leq n$ and $\bfA$ has at most $p+k$ pairwise-distinct rows and at most $q+k$ distinct columns.
Further, we exhaustively apply the following reduction rule.

\begin{redrule}
If $\bfA$ contains at least $k+p+1$ rows that are pairwise equal, then delete one of these rows.
\end{redrule}

To see that the rule is safe, let $\bfB$ be an $m\times n$-matrix such that $\|\bfA-\bfB\|_F^2\leq k$. If for $I\subseteq\{1,\ldots,m\}$ it holds that $|I|\geq p+k+1$ and the rows $\bfa_i$ for $i\in I$ are the same, then there is $J\subseteq I$ of size at least $p+1$ such that the rows $\bfb_i $ of $\bfB$ for $i\in J$ are the same. Then $\bfB$ is a $\bfP$-matrix if and only if the matrix obtained by the deletion of an arbitrary row $\bfb_i$ for $i\in J$ is a $\bfP$-matrix.

For simplicity, we keep the same notations for the matrix  obtained from $\bfA$ by the exhaustive application of the rule, i.e., we assume that $\bfA$ is an $m\times n$-matrix with the columns $\bfa^1,\ldots,\bfa^n$. Note that now we have that $m\leq (p+k)^2$.

Let $\bfB=(\bfb^1,\ldots,\bfb^n)$ be a solution to  $(\bfA,\bfP,k)$. We say that a set $J\subseteq\{1,\ldots,n\}$ \emph{represents} $\bfP$ with respect to $\bfB$ if a) $|J|=q$, b) $\bfB[\{1,\ldots,m\},J]$ is a $\bfP$-matrix, c) for each $j\in \{1,\ldots,n\}\setminus J$, $\bfb^j=\bfb^i$ for some $i\in J$, and d) the value of $\|\bfA[\{1,\ldots,m\},J]-\bfB[\{1,\ldots,m\},J]\|_F^2$ is minimum over all $J\subseteq\{1,\ldots,n\}$ satisfying a)--c). Observe that for every solution $\bfB$, there is $J\subseteq\{1,\ldots,n\}$ that represents $\bfP$ with respect to $\bfB$. 

We say that two sets of indices $J,J'\subseteq \{1,\ldots,n\}$ are \emph{equivalent} with respect to $\bfA$ if the matrices $\bfA[\{1,\ldots,m\},J]$ and $\bfA[\{1,\ldots,n\},J']$ are isomorphic. 
Observe that if $J,J'$ are equivalent with respect to $\bfA$, then $(\bfA,\bfP,k)$ has a solution $\bfB$ such that $J$ represents $\bfP$ with respect to $\bfB$ if and only if the same instance has a solution 
$\bfB'$ such that $J'$ represents $\bfP$ with respect to $\bfB'$. 

Because $\bfA$ has at most $q+k$ pairwise-distinct columns, there are at most $(q+k)^q$  sets of indices $J\in\{1,\ldots,n\}$ of size $q$ which are pairwise nonequivalent with respect to $\bfA$. We consider such sets and for each $J$, we check whether there is a solution $\bfB$ for  $(\bfA,\bfP,k)$ with $J$ representing $\bfP$ with respect to $\bfB$. If we find that there is a solution, we return YES and stop, and we return NO if there is no solution for every choice of $I$. From now we assume that $J$ is given.

We construct the instance $(\bfA,\bfP,\{X,Y,Z\},k)$ of \probPClustExt with $X=\emptyset$, $Y=J$ and $Z=\{1,\ldots,n\}\setminus J$. We have that $\bfB$ is a solution to  $(\bfA,\bfP,\{X,Y,Z\},k)$ of \probPClustExt if and only if $\bfB$ is a solution  for the instance $(\bfA,\bfP,k)$ of \probPClust with $J$ representing $\bfP$ with respect to $\bfB$. 
Therefore, in order to solve \probPClust it suffices to 
 solve \probPClustExt.

We construct the recursive branching algorithm for \probPClustExt called \textsc{Extend-$\bfP$-Solution}. The algorithm takes as an  input  matrix $\bfA$, disjoint sets of indices $X,Y,Z\subseteq \{1,\ldots,n\}$ such that $|X|+|Y|=q$,  and a nonnegative integer $k$.
\textsc{Extend-$\bfP$-Solution} executes in three steps.

\medskip
\noindent
{\bf Step~1.} If $Y=\emptyset$, then do the following.
\begin{itemize}
\item[(i)] Check whether $\bfA[\{1,\ldots,m\},X]$ is a $\bfP$-matrix  using Observation~\ref{obs:p-matr-brute}. If it's not, then return NO and stop.
\item[(ii)] If $Z=\emptyset$, then return YES and stop.
\item[(iii)] Verify whether
\begin{equation}\label{eq:Z}
\sum_{i\in Z}\min\{\hdist(\bfa^i,\bfa^j)\mid j\in X\}\leq k,
\end{equation}
return YES if (\ref{eq:Z}) holds, return NO  otherwise;  then  stop.
\end{itemize}

\medskip
\noindent
{\bf Step~2.} For every $h=0,\ldots,k$, do the following.
\begin{itemize}
\item[(i)]  For every $i\in Z$, if $\ell=\min\{\hdist(\bfa^i,\bfa^j)\mid j\in X\}\leq h$, then set $Z=Z\setminus\{i\}$ and $k=k-\ell$.
\item[(ii)] If $|Y|+|Z|\leq \sqrt{k\log (p+k)}$, then solve $(\bfA[\{1,\ldots,m\},X\cup Y\cup Z],\bfP,\{X,Y,Z\},k)$ using Lemma~\ref{lem:ext-Pclust}.
\item[(iii)] If $h\leq k/(|Y|+|Z|)$, then for each $i\in Y$ and each vector $\hat{\bfa}^i\in\{0,1\}^m$ such that $\hdist(\bfa^i,\hat{\bfa}^i)= h$,  let $\hat{\bfA}$ be the matrix obtained from $\bfA$ by t replacing   column $\bfa^i$ by $\hat{\bfa}^i$, call  \textsc{Extend-$\bfP$-Solution}$(\hat{\bfA},X\cup\{i\},Y\setminus \{i\},Z,k-h)$  and if the algorithm returns YES, then return YES and stop.
\end{itemize}

\medskip
\noindent
{\bf Step~3.} Return NO and stop.

\medskip
We call \textsc{Extend-$\bfP$-Solution}$(\bfA,\emptyset,J,\{1,\ldots,n\}\setminus J,k)$ to solve the instance of \probPClustExt that was constructed above for the considered set of indices $J$.

\medskip\noindent\textbf{Correctness.}
To prove correctness, we show first that if the algorithm returns  YES, then $(\bfA[\{1,\ldots,m\},X\cup Y\cup Z],\bfP,\{X,Y,Z\},k)$ is a \yesinstance of \probPClustExt. 

Suppose that $Y=\emptyset$. This case is considered in Step~1. If the algorithms returns YES, then it does not stop in Step~1 (i). Hence, $\bfA[\{1,\ldots,m\},X]$ is a $\bfP$-matrix. If $Z=\emptyset$, we have that $\bfB=\bfA[\{1,\ldots,m\},X]$ is a solution for the instance  $(\bfA[\{1,\ldots,m\},X\cup Y\cup Z],\bfP,\{X,Y,Z\},k)$ of \probPClustExt and the algorithm correctly returns YES. Let $Z\neq \emptyset$.
For every $i\in X$, we set $\bfb^i=\bfa^i$, and for each $i\in Z$, we define the vector $\bfb^i=\bfa^h$ for $h\in X$, where $\hdist(\bfa^i,\bfa^h)=\min\{\hdist(\bfa^i,\bfa^j)\mid j\in X\}$. Consider  matrix $\bfB$ composed of the columns $\bfb^i$ for $i\in X\cup Y\cup Z$. It is easy to verify that $\bfB$ is a solution to $(\bfA[\{1,\ldots,m\},X\cup Y\cup Z],\bfP,\{X,Y,Z\},k)$. Hence, if the algorithm returns YES in Step~1 (iii), then $(\bfA[\{1,\ldots,m\},X\cup Y\cup Z],\bfP,\{X,Y,Z\},k)$ is a \yesinstance of \probPClustExt. 
 
Let $Y\neq \emptyset$. Assume that the algorithm returns YES for some $h\in\{0,\ldots,k\}$ in Step~2. 

Denote by $Z'$ the set obtained from $Z$  in Step~2 (i) and let $k'$ be the value of $k$ obtained in Step~2 (i). Observe that if  $(\bfA[\{1,\ldots,m\},X\cup Y\cup Z'],\bfP,\{X,Y,Z'\},k')$ is a \yesinstance of \probPClustExt, then $(\bfA[\{1,\ldots,m\},X\cup Y\cup Z],\bfP,\{X,Y,Z\},k)$ is a \yesinstance as well. Indeed,  let  $\hat{\bfB}$ be  a solution to   $(\bfA[\{1,\ldots,m\},X\cup Y\cup Z'],\bfP,\{X,Y,Z'\},k')$ with the columns $\hat{\bfb}^i$ for $i\in X\cup Y\cup Z'$.
We define $\bfB$ with the columns $\bfb^i$ for $i\in X\cup Y\cup Z$ as follows: for $i\in X\cup Y\cup Z'$, $\bfb^i=\hat{\bfb}^i$, and for $i\in Z\setminus Z'$, $\bfb^i=\bfa^h=\hat{\bfb}^h$ for $h\in X$ such that 
$\hdist(\bfa^i,\bfa^h)=\min\{\hdist(\bfa^i,\bfa^j)\mid j\in X\}$. It is easy to see that $\bfB$ is a solution for $(\bfA[\{1,\ldots,m\},X\cup Y\cup Z],\bfP,\{X,Y,Z\},k)$.

Suppose that the algorithm returns YES in Step~2 (ii). Then $(\bfA[\{1,\ldots,m\},X\cup Y\cup Z'],\bfP,\{X,Y,Z'\},k')$,  and therefore $(\bfA[\{1,\ldots,m\},X\cup Y\cup Z],\bfP,\{X,Y,Z\},k)$,  are \yesinstances of \probPClustExt. 

Suppose that the algorithm returns YES in Step~2 (iii) for $i\in Y$. Then $(\hat{\bfA}[\{1,\ldots,m\},X\cup Y\cup Z],\bfP,\{X\cup\{i\},Y\setminus\{i\},Z'\},k')$ is a \yesinstance of \probPClustExt. 
Since $\hdist(\bfa^i,\hat{\bfa}^i)= h$, we have that every solution to  $(\hat{\bfA}[\{1,\ldots,m\},X\cup Y\cup Z],\bfP,\{X\cup\{i\},Y\setminus\{i\},Z'\},k')$ is also a  solution to 
 $(\bfA[\{1,\ldots,m\},X\cup Y\cup Z],\bfP,\{X,Y,Z'\},k')$. This implies that $(\bfA[\{1,\ldots,m\},X\cup Y\cup Z'],\bfP,\{X,Y,Z'\},k')$ and, therefore, $(\bfA[\{1,\ldots,m\},X\cup Y\cup Z],\bfP,\{X,Y,Z\},k)$ is a \yesinstance of \probPClustExt. 

We proved that if the algorithm returns the answer YES, then $(\bfA[\{1,\ldots,m\},X\cup Y\cup Z],\bfP,\{X,Y,Z\},k)$ is a \yesinstance of \probPClustExt. Recall that we 
call \textsc{Extend-$\bfP$-Solution}$(\bfA,\emptyset,J,\{1,\ldots,n\}\setminus J,k)$ to solve the instance of \probPClustExt  constructed for the considered set of indices $J\subseteq\{1,\ldots,n\}$ of size $q$ from the instance $(\bfA,\bfP,k)$ of \probPClust. It follows that if \textsc{Extend-$\bfP$-Solution}$(\bfA,\emptyset,J,\{1,\ldots,n\}\setminus J,k)$ outputs YES, then $(\bfA,\bfP,k)$ is a \yesinstance of  
\probPClust.

Let us note  that \textsc{Extend-$\bfP$-Solution}$(\bfA,X,Y,Z,k)$ can return NO  even if \linebreak  $(\bfA[\{1,\ldots,m\},X\cup Y\cup Z],\bfP,\{X,Y,Z\},k)$ is a \yesinstance of \probPClustExt. This can occur when  $X$, $Y$ and $Z$ are arbitrary disjoint subsets of $\{1,\ldots,n\}$ such that $|X|+|Y|=q$. Nevertheless, because we call 
call \textsc{Extend-$\bfP$-Solution}$(\bfA,\emptyset,J,\{1,\ldots,n\}\setminus J,k)$, we are able to show the following claim. 

\begin{claim}\label{claim:correct}
If $(\bfA,\bfP,k)$ is a \yesinstance of  \probPClust and a set $J\subseteq\{1,\ldots,n\}$ of size $q$ is selected  such   that $(\bfA,\bfP,k)$ has a solution $\bfB$ and  
  $J$ represents $\bfP$ with respect to $\bfB$, then \textsc{Extend-$\bfP$-Solution}$(\bfA,\emptyset,J,\{1,\ldots,n\}\setminus J,k)$ returns YES.
\end{claim}

\begin{proof}[Proof of Claim~\ref{claim:correct}] By making use of induction on $|Y|$, we show   that \linebreak \textsc{Extend-$\bfP$-Solution}$(\hat{\bfA},X,Y,Z,k')$ returns YES if 
  $(\hat{\bfA}[\{1,\ldots,m\},X\cup Y\cup Z],\bfP,\{X,Y,Z\},k')$ is a \yesinstance of \probPClustExt,  where
\begin{itemize}
\item  $X\cup Y=J$, 
\item  $\hat{\bfA}$ is the matrix with the columns $\hat{\bfa}^i$ for $i\in X\cup Y\cup Z$, where $\hat{\bfa}^i=\bfb^i$ for $i\in X$ and $\hat{\bfa}^i=\bfa^i$ for $i\in Y\cup Z$, and 
\item
$k'\geq k-\sum_{i\in X\cup (\{1,\ldots,n\}\setminus (X\cup Y\cup Z))}\hdist(\bfa^i,\bfb^i)$. 
 \end{itemize}
 Note that for $X=\emptyset$, this would imply the claim.

The base of the induction is the case $|Y|=0$, i.e., $Y=\emptyset$. This case is considered in Step~1 of  \textsc{Extend-$\bfP$-Solution}. Since $\hat{\bfA}[\{1,\ldots,m\},X]=\bfB[\{1,\ldots,m\},X]$, we have that  $\bfB[\{1,\ldots,m\},X]$ is a $\bfP$-matrix. In particular, the algorithm does not stop in Step~1 (i). If $Z=\emptyset$, then the algorithm returns YES. Let $Z\neq\emptyset$. 
Notice that because 
$\hat{\bfA}[\{1,\ldots,m\},X]=\bfB[\{1,\ldots,m\},X]$ and $J$ represents $\bfP$ with respect to $\bfB$, we have that 
$$\sum_{i\in Z}\min\{\hdist(\bfa^i,\bfa^j)\mid j\in X\}\leq \sum_{i\in Z}\hdist(\bfa^i,\bfb^i)\leq k'.$$
Hence, the algorithm returns YES in Step~2 (iii).

Assume that $|Y|>0$ and we proved our statement for smaller sets $Y$. Let $h^*=\min\{\hdist(\bfa^i,\bfb^i)\mid i\in Y\}$. Let also $i^*\in Y$ be such that $h^*=\hdist(\bfa^{i^*},\bfb^{i^*})$.
We claim that  \textsc{Extend-$\bfP$-Solution} returns YES in Step~3 for $h=h^*$ unless it does not  return YES before. 

Denote by $Z^*$ the set obtained from $Z$ in Step~2 (i) and let $k^*$ the value of $k$ obtained in the same step for $h=h^*$. 

Recall that $J$ represents $\bfP$ with respect to $\bfB$. By the condition d) of the definition, we have that the value of $\|\bfA[\{1,\ldots,m\},J]-\bfB[\{1,\ldots,m\},J]\|_F^2$ is minimum over all $J\subseteq\{1,\ldots,n\}$ satisfying a)--c). In particular, it implies that for each $i\in Z$, either $\bfb^i=\bfb^j$ for some $j\in X$ or $\hdist(\bfa^i,\bfb^i)\geq h^*$. It follows that 
$(\hat{\bfA}[\{1,\ldots,m\},X\cup Y\cup Z^*],\bfP,\{X,Y,Z\},k^*)$ is a \yesinstance of \probPClustExt and $k^*\geq k-\sum_{i\in X\cup (\{1,\ldots,n\}\setminus (X\cup Y\cup Z^*))}\hdist(\bfa^i,\bfb^i)$.

If $|Y|+|Z^*|\leq \sqrt{k^*\log (p+k^*)}$, we solve \probPClustExt for \linebreak $(\bfA[\{1,\ldots,m\},X\cup Y\cup Z^*],\bfP,\{X,Y,Z^*\},k^*)$ directly using Lemma~\ref{lem:ext-Pclust}. Hence, the algorithm returns YES.

Assume that $|Y|+|Z^*|> \sqrt{k^*\log(p+k^*)}$. Notice that for each $i\in Z^*$, we have that $\hdist(\bfa^i,\bfb^i)\geq h^*$. Hence, 
$$k^*\geq \sum_{i\in Y\cup Z}\hdist(\bfa^i,\bfb^i)\geq h^* \cdot (|Y|+|Z|)$$
and $h^*\leq k^*/(|Y|+|Z|)\leq \sqrt{k^*/\log (p+k^*)}$. In Step~2 (iii), we consider every  $i\in Y$ and each vector $\hat{\bfa}^i\in\{0,1\}^m$ such that $\hdist(\bfa^i,\hat{\bfa}^i)= h^*$. In particular, we consider $i=i^*$ and $\hat{\bfa}^{i^*}=\bfb^{i^*}$. Then $(\hat{\bfA}[\{1,\ldots,m\},X\cup Y\cup Z^*],\bfP,\{X\cup\{i^*\},Y\setminus\{i^*\},Z^*\},k^*-h^*)$ for the matrix $\hat{\bfA}^*$ 
 obtained from $\hat{\bfA}$ by the replacing   column $\hat{\bfa}^{i^*}=\bfa^{i^*}$ by $\bfb^{i^*}$ is a \yesinstance of \probPClustExt. By the inductive assumption, we have that the algorithms returns YES for this branch. Recall that whenever the algorithm in Step~2 (iii) obtains YES for some branch, it returns YES and stops.  Since we have such an answer for at least one branch, the algorithm returns YES.
This concludes the proof of Claim.
\end{proof}

Summarizing, we obtain that $(\bfA,\bfP,k)$ is a \yesinstance of  \probPClust if and only if \textsc{Extend-$\bfP$-Solution}$(\bfA,\emptyset,J,\{1,\ldots,n\}\setminus J,k)$ returns YES for some choice of $J\subseteq\{1,\ldots,n\}$ of size $q$. This competes the correctness proof.

\medskip\noindent\textbf{Running time.}
Now we evaluate the running time. The preprocessing is done in polynomial time. Then we consider all  pairwise nonequivalent with respect to $\bfA$ sets of indices $J\in\{1,\ldots,n\}$ of size $q$. There are at most $(q+k)^q$ such sets. Then for each $J$, we run \textsc{Extend-$\bfP$-Solution}$(\bfA,\emptyset,J,\{1,\ldots,n\}\setminus J,k)$. By Observation~\ref{obs:p-matr-brute}, 
Step~1 (i) can be done in time  $2^{\Oh(p\log p+q\log q)}\cdot (nm)^{\Oh(1)}$. 
Parts (ii) and (iii) of Step~1 are performed in polynomial time.
In Step~2, we consider $k+1$ values of $h$ and for each $h$ perform Step 2 (i)--(iii). Step~2 (i) takes polynomial time. Step~2 (ii) can be done in time  $2^{\Oh(p(\sqrt{k\log (p+k)}+\log k)+p\log p+q\log q)}\cdot (nm)^{\Oh(1)}$ by Lemma~\ref{lem:ext-Pclust}.
In Step~2 (iii), we consider at most $|Y| \leq q$ values of $i$, and for each $i$ construct all vectors  $\hat{\bfa}^i\in\{0,1\}^m$ such that $\hdist(\bfa^i,\hat{\bfa}^i)= h\leq \sqrt{k/\log (p+k)}$. Recall that after the preprocessing, we have that $m\leq(p+k)^2$. Hence, we have at most $(p+k)^{2\sqrt{k/\log (p+k)}}$ vectors $\hat{\bfa}^i$. It means that in Step~2 (iii) we have $2^{\Oh(\sqrt{k\log(p+k)})}$ branches. On each recursive call, we reduce the size of $Y$. It means that the depth of the recursion is at most $q$. 
Then the total running time is \[2^{\Oh(p\sqrt{k\log (p+k)}+q\log\sqrt{k\log (p+k)}+p\log p+q\log q+q\log(q+k))}\cdot (nm)^{\Oh(1)}.\]
\end{proof}

Note that the running time in Theorem~\ref{thm:pclust-subexp} is asymmetric in $p$ and $q$ due to the fact that we treat rows and columns in different way but, trivially, the instances $(\bfA,\bfP,k)$ and $(\bfA^\intercal,\bfP^\intercal,k)$ of \probPClust are equivalent.
 If $p$ and $q$ are assumed to be constants,  then \probPClust is solvable in time 
 $2^{\Oh(\sqrt{k\log k})}\cdot (nm)^{\Oh(1)}$.

Notice that we can invoke  Theorem~\ref{thm:pclust-subexp} to solve \probFact as follows.
 We use Observation~\ref{obs:rank} and observe that $(\bfA,r,k)$ is a \yesinstance of \probFact if and only if there is a $2^r\times 2^r$-matrix $\bfP$ of \GF-$\rank $ at most $r$ such that $(\bfA,\bfP,k)$ is a \yesinstance of \probPClust. 
 Matrix $\bfP$ is of  \GF-$\rank$  $r$ if and only if it can be represented as a product $\bfP=\bfU \cdot \bfV$, 
 where  $\bfU$ is $2^r\times r$ and $ \bfV$ is $r\times 2^r$ binary matrix and arithmetic operations are over \GF. 
 There are at most $2^{r2^{r}}$ different 
 binary  $2^r\times r$-matrices $\bfU$,   and at most $2^{r2^{r}}$ different  binary  $r\times 2^r$ -matrices $\bfV$.  Thus there are at most  $2^{2r2^{r}}$  candidate  matrices $\bfP$.
 For each such matrix $\bfP$, we check whether   $(\bfA,\bfP,k)$ is a \yesinstance of \probPClust by invoking Theorem~\ref{thm:pclust-subexp}. However this approach gives a double exponential dependence in $r$, which is much worse the bound provided by   Theorem~\ref{thm:fact-subexp}. Still, this approach is useful if we consider the variant of \probFact for Boolean matrices.

Let us remind that  binary matrix $\bfA$ has the Boolean rank $1$ if $A=\bfx\wedge \bfy^\intercal$ where $\bfx\in\{0,1\}^m$ and $\bfy\in\{0,1\}^n$ are nonzero vectors and the product is  Boolean and that 
  the Boolean rank of $\bfA$ is the minimum integer $r$ such that $\bfA=\bfA^{(1)}\vee\cdots\vee \bfA^{(r)}$ where $\bfA^{(1)},\ldots,\bfA^{(r)}$  are matrices of Boolean rank 1 and the sum is Boolean.

\begin{theorem}\label{cor:brank}
\probBFact   is solvable  in  $2^{\Oh(r2^r\sqrt{k\log k})}\cdot(nm)^{\Oh(1)}$ time.
\end{theorem}

\begin{proof}
Let $\bfA$ be a Boolean $m\times n$-matrix with the Boolean rank $r\geq 1$. Then  $\bfA=\bfA^{(1)}\vee \ldots\vee \bfA^{(r)}$ where $\bfA^{(1)},\ldots,\bfA^{(r)}$  are matrices of Boolean rank 1. It implies that $\bfA$ has at most $2^r$ pairwise-distinct rows and at most $2^r$ pairwise-distinct columns. Hence, the Boolean rank of $\bfA$ is at most $r$ if and only if there is a $2^r\times 2^r$-matrix $\bfP$ of Boolean rank at most $r$ such that $\bfA$ is a $\bfP$-matrix. Respectively,  the \probBFact problem can be reformulated as follows: 
Decide whether there is a $2^r\times 2^r$-pattern matrix $\bfP$ with the Boolean rank at most $r$ and an $m\times n$ $\bfP$-matrix $\bfB$ such that $\|\bfA-\bfB\|_F^2\leq k$.

We generate all  $2^r\times 2^r$-matrices $\bfP$ of Boolean rank at most $r$, and then for each matrix $\bfP$, we solve \probPClust for the instance $(\bfA,\bfP,k)$. We return YES if we obtain at least one \yesinstance of \probPClust, and we return NO otherwise. 
 
 By definition, the {Boolean rank} of $\bfP$ is $r$ if and only if  $\bfP=\bfU\wedge \bfV$ for a Boolean $2^r\times r$ matrix $\bfU$ and a Boolean $r\times 2^r$ matrix $\bfV$
Since there are at most $2^{r2^{r}}$ $2^r\times r$-matrices, we construct all the $2^r\times 2^r$-matrices $\bfP$ with the Boolean rank at most $r$ in time $2^{\Oh(r2^{r})}$.
By Theorem~\ref{thm:pclust-subexp},    the considered instances of \probPClust is solvable in time  $2^{\Oh(r2^r\sqrt{k\log k})}\cdot(nm)^{\Oh(1)}$.
\end{proof}

In the conclusion of the section we observe that we hardly can avoid the double exponential dependence on $r$ for \probBFact. Chandran, Issac and Karrenbauer proved in~\cite{ChandranIK16}
that the \textsc{Biclique Cover} problem that asks, given a bipartite graph $G$, whether the set of edges of $G$ could be covered by at most $r$ bicliques (that is, complete bipartite graphs) cannot be solved in time $2^{2^{o(r)}}\cdot |V(G)|^{\Oh(1)}$ unless the \emph{Exponential Time Hypothesis} (ETH) is false (we refer to~\cite{CyganFKLMPPS15} for the introduction to the algorithmic lower bounds based on ETH). Since \textsc{Biclique Cover} is equivalent to deciding whether the bipartite adjacency matrix of $G$ has the Boolean rank at most $r$, \probBFact cannot be solved in time 
$2^{2^{o(r)}}\cdot(nm)^{\Oh(1)}$ for $k=0$ unless ETH fails.

\section{Conclusion and open problems}\label{sec:conclusion}
In this paper we provide a number of parameterized algorithms for a number of binary matrix-approximation problems. Our results uncover some parts of the  complexity landscape of these fascinating  problems. We hope that our work will facilitate further investigation of this important and exciting area. We   conclude with the following concrete open problems about bivariate complexity of \probClust, \probFact, and \probBFact.  

For \probClust we have shown that the problem is solvable in time $2^{\Oh(k\log k)}\cdot(nm)^{\Oh(1)}$. A natural question is whether this  running time is optimal. While  the lower bound of the kind
$2^{o(k)}\cdot(nm)^{\Oh(1)}$ or $2^{o(k\log{k})}\cdot(nm)^{\Oh(1)}$ seems to be  most plausible here,  we do not know any strong argument against, say   a  $2^{o({k})}\cdot(nm)^{\Oh(1)}$-time algorithm. At least for the number of distinct columns  $r\in \Oh(k^{1/2 -\varepsilon})$ with $\varepsilon>0$,
we have a subexponential in $k$ algorithm, so maybe we can solve the problem in time subexponential in $k$  for any value of $r$?

For \probFact we have an algorithm solving the problem in time   $2^{\Oh(r^{ 3/2}\cdot \sqrt{k\log{k}})}(nm)^{\Oh(1)}$. Here, shaving off the $\sqrt{\log{k}}$ factor in  the exponent seems to be a reasonable thing.   However, we do not know how to do it even by the cost of the worse dependence in $r$. In other words, could the problem be solvable in time  $2^{\Oh(f(r)\cdot \sqrt{k})}(nm)^{\Oh(1)}$ for some function $f$? On the other hand, we also do not know how to rule out  algorithms running  in  time $2^{o(r)\cdot o(k))}(nm)^{\Oh(1)}$.
 
For \probBFact, how far is our upper bound $2^{\Oh(r2^r\cdot \sqrt{k\log k})}(nm)^{\Oh(1)}$  from the optimal?  For example, we know that for any function $f$,  the solvability of the problem in  time $2^{2^{o(r)}} f(k) (nm)^{\Oh(1)}$ implies the failure of ETH. Could we rule out  $2^{o(\sqrt{k})} f(r) (nm)^{\Oh(1)}$ algorithm?

From kernelization perspective, we proved that  \probClust  admits a polynomial kernel when parameterized by $k+r$. On the other hand, due to its connection to  \textsc{Biclique Cover}, we  know that already for $k=0$,  \probBFact does not admit a subexponential kernel when parameterized by $r$~\cite{ChandranIK16}. This rules out the existence of  a polynomial in $r+k$ kernel for  \probBFact. However, for  \probFact the existence of a polynomial in $r+k$ kernel is open.

\paragraph{Acknowledgments.} We thank Daniel Lokshtanov, Syed Mohammad Meesum and Saket Saurabh for helpful discussions on the topic of the paper.

\end{document}